\documentclass[12pt]{article}

\usepackage{fullpage}
\usepackage[T1]{fontenc}
\usepackage{ae,aecompl}
\usepackage[dvips]{graphicx}
\usepackage{amssymb}
\usepackage{amsmath}
\usepackage{epsf}
\usepackage{color}
\usepackage{array}
\usepackage{rotating}
\usepackage{colortbl}
\definecolor{webgreen}{rgb}{0,0.4,0}
\definecolor{webbrown}{rgb}{0.6,0,0}
\definecolor{purple}{rgb}{0.5,0,0.25}
\definecolor{darkblue}{rgb}{0,0,0.7}
\definecolor{darkred}{rgb}{0.7,0,0}
\usepackage[dvips,pdfborder=false]{hyperref}
\hypersetup{colorlinks,citecolor=darkred,filecolor=black,linkcolor=darkblue,urlcolor=webgreen,pdfpagemode=None,
pdfstartview=FitH}
\newcommand{\ignore}[1]{}
\newtheorem{lemma}{{\sc Lemma}}
\newtheorem{prop}{{\sc Proposition}}

\newtheorem{theorem}{{\sc Theorem}}
\newtheorem{defn}{{\sc Definition}}

\newtheorem{claim}{{\sc Claim}}
\newtheorem{example}{{\sc Example}}
\newtheorem{fact}{{\sc Fact}}

\newenvironment{proof}{\noindent {\bf \sl Proof\/}:\enspace}
{\hfill $\blacksquare{}$ \vspace{12pt}}

\begin{document}

\title{{\sc Roberts' Theorem with Neutrality: \\ A Social Welfare Ordering Approach}~\thanks{
We are extremely grateful to Sushil Bikhchandani, Juan Carlos Carbajal, Claude d'Aspremont,  
Mridu Prabal Goswami, Amit Goyal, Ron Lavi, Thierry Marchant, Herv\'{e} Moulin, Arup Pal, David Parkes, Kevin Roberts, 
Maneesh Thakur, Rakesh Vohra, and
seminar participants at the ``Multidimensional Mechanism Design" workshop
in Bonn, Universitat Aut\`{o}noma de Barcelona, and Delhi School
of Economics for their comments.}}
\author{Debasis Mishra\thanks{Corresponding Author. Indian
Statistical Institute, 7, S.J.S. Sansanwal Marg, New Delhi - 110016.}~ and Arunava Sen~\thanks{Indian
Statistical Institute, 7, S.J.S. Sansanwal Marg, New Delhi - 110016.}}
\maketitle

\begin{abstract}
We consider dominant strategy implementation in private values settings, when agents have 
multi-dimensional types, 
the set of alternatives is finite, monetary
transfers are allowed, and agents have
quasi-linear utilities. We show that any implementable and neutral social choice function
must be a weighted welfare maximizer if the type space of every agent is an
$m$-dimensional open interval, where $m$ is the number of alternatives.
When the type space of every agent is unrestricted, Roberts' theorem with neutrality \cite{Roberts79}
becomes a corollary to our result. Our proof technique uses a {\em social welfare ordering}
approach, commonly used in aggregation literature in social choice theory. We 
also prove the general (affine maximizer) 
version of Roberts' theorem for unrestricted type spaces of agents using this approach.
\end{abstract}



\newpage

\section{Introduction}

The well-known Gibbard-Satterthwaite \cite{Gibbard73,Satter75} Impossibility Theorem in
mechanism design asserts that in unrestricted domains, every implementable social choice
function which has at least three alternatives in its range, must be dictatorial. A crucial
aspect of the unrestricted domain assumption is that monetary transfers are not permitted.
 However, models where monetary transfers are admissible are very important. Both the auction setting
 and the standard public good model assume that agents can receive monetary transfers (either positive
 or negative) and that the underlying utility function of every agent is quasi-linear in money. This
 paper is a contribution to the literature which investigates the structure of social choice functions
 which can be implemented in dominant strategies in these settings.

\cite{Vickrey61, Clarke71, Groves73} showed that efficient social choice functions can be
implemented by a unique family of transfer rules, now popularly known as Vickrey-Clarke-Groves
(VCG) transfer schemes. Remarkably, when the domain is unrestricted (as in the Gibbard-Satterthwaite
setup) and the range of the mechanism contains at least three alternatives, the
only (dominant strategy) implementable social choice functions are {\em affine maximizers}.
These social choice functions are generalizations of weighted efficiency rules. 
This result was proved by
\cite{Roberts79} in a seminal paper. It can be seen as the counterpart to
the Gibbard-Satterthwaite theorem for quasi-linear utility environments.

As in the literature without money, the literature with quasi-linear utility has since tried to
relax various assumptions in Roberts' theorem. \cite{Rochet87} shows that a certain
{\em cycle monotonicity} property characterizes dominant strategy implementable
social choice functions. Though this characterization is very general - works for any domains and any
set of alternatives (finite or infinite) - it is not as useful as the Roberts' theorem since
it does not give a functional form of the class of implementable social choice functions. Along the
lines of \cite{Rochet87}, \cite{Bikh06} and \cite{Saks05} have shown that a
{\em weak monotonicity} property characterizes implementable social choice
functions in auction settings, a severely restricted domain, when the set of alternatives is finite and the type space
is convex~\footnote{See also \cite{Monderer08} and \cite{Ashlagi09}. \cite{Ashlagi09}
prove the converse of this result also.}.
Again, the precise functional form of the implementable social choice functions
are missing in these characterizations. A fundamental open question is the following: \\

\noindent {\em What subdomains allow for a functional form of
implementable social choice functions?} \\

Several attempts have been made recently to simplify, refine, and extend
Roberts' theorem.
Using almost the same structure
and approach, \cite{Lavi08} reduced the complexity of Roberts' original proof. 
\cite{Dobzinski09} also provide an alternate (modular) proof of Roberts' theorem
for unrestricted domain. Building on the technique of \cite{Lavi08}, \cite{Carbajal09} extend
Roberts' theorem to {\em continuous} domains. Other proofs of Roberts' theorem can be found in
(for unrestricted domains) \cite{Jehiel08} and \cite{Vohra08}.

\subsection{Our Contribution}

Our paper contributes to the literature in two ways. First, we characterize  restricted domains where
the affine maximizer theorem holds
in the presence of an additional assumption on social choice functions, that of {\it neutrality}.
Neutrality requires the social choice function to treat all alternatives symmetrically.
It is a familiar axiom in social choice theory and we discuss it at greater length in Section
\ref{subsection:neutrality}. Our main result states that every implementable social choice
function is a weighted welfare maximizer if the type space of every agent is an $m$-dimensional
open interval, where $m$ is the number of alternatives.
For the unrestricted domain, our result implies Roberts' result in the special case where attention is
restricted to neutral social choice functions. We demonstrate that the neutrality assumption is essential for our
domain characterization result in the following sense: there exist (open interval) domains over which
an implementable non-affine-maximizer social choice function exists but over which all
neutral implementable social choice functions are weighted welfare maximizers.

Our second contribution is methodological and conceptual. Our proof technique differs significantly
from existing ones.  It can be summarized in three steps.

\begin{itemize}
\item[S1] We show that an implementable and neutral social choice function induces an
ordering on the domain.

\item[S2] This ordering satisfies three key properties: {\em weak Pareto, invariance, and continuity}.

\item[S3] We then prove a result on the representation of any ordering which satisfies these properties.
For unrestricted domains this result is familiar in the literature - see for instance, \cite{Blackwell54}, \cite{despremont77},
\cite{Blackorby84}, \cite{Trockel92} and \cite{despremont02}. We show that any ordering
on an open and convex set which satisfies the axioms specified in S2 can be represented
by a weighted welfare maximizer.

\end{itemize}

The key feature of our approach is to transform the problem of characterizing incentive-compatible social
choice functions over a domain into a particular problem of characterizing orderings of vectors in that domain.
The problem of characterizing orderings satisfying properties such as weak Pareto, invariance, continuity
etc (over the unrestricted domain), is a classical one in social choice theory. It arose from the recognition
that a natural way to escape the negative conclusions of the Arrow Impossibility Theorem was to enrich the
informational basis of Arrovian social welfare functions from individual preference orderings to utility
functions. If a social welfare function satisfies the (standard) axioms of Independence of Irrelevant
Alternatives and Pareto Indifference, then it is ``equivalent'' to an ordering over $\mathbb{R}^n$ where $n$
is the number of individuals. The aggregation problem in this environment can therefore be reduced to the
problem of determining an appropriate ordering of the vectors in $\mathbb{R}^n$. There is an extensive literature
which investigates exactly this question (see \cite{despremont02} for a comprehensive survey).

It has been known that there is a deep connection between two seemingly unrelated problems in social
choice - the strategic problem with the goal of characterizing incentive-compatible social choice functions and the
aggregation problem with the objective of characterizing social welfare functions satisfying the Arrovian axioms.
For instance in the case of the unrestricted domain consisting of all preference orderings, the Arrow Theorem
can be used to prove the Gibbard-Satterthwaite Theorem and vice-versa (for a unified approach to both problems
see \cite{Reny01}). Our proof serves to highlight this connection further by demonstrating the equivalence of
a strategic problem in a quasi-linear domain with an aggregation problem involving utility functions.

We also remark that though the representation result in Step S3 is well-known for unrestricted domains, 
and our extension to
open and convex domain may be of some independent interest.

Finally, we show how Roberts' affine maximizer theorem cab be proved using Roberts' theorem with neutrality.
This proof is contained in Section \ref{sec:gen}. 

\section{Roberts' Affine Maximizer Theorem}

Let $A=\{a,b,c,\ldots\}$ be a finite set of alternatives or allocations.
Suppose $|A|=m \ge 3$. Let $N=\{1,\ldots,n\}$ be a finite
set of agents. The type of agent $i$ is a vector in $\mathbb{R}^m$.
Denote by $t_i$ the type (vector) of agent $i \in N$, where
for every $a \in A$, $t_i^a$ denotes the {\em value} of
agent $i$ for alternative $a$ when his type is $t_i$.
A type profile will be denoted by $t$, and consists of
$n$ vectors in $\mathbb{R}^m$. Alternately, one
can view a type profile $t$ to be an $n \times m$ matrix,
where every row represents a type vector of an agent.
The column vectors are vectors in $\mathbb{R}^n$.
We refer to a column vector generated by a type profile
to be a utility vector. Hence, $t^a$ represents the utility
vector corresponding to allocation $a$ in type profile
$t$ and $t^{-a}$ will denote the
utility vectors in type profile $t$ except $t^a$. 

Let $T_i$ be the type space (the set of all type vectors) of
agent $i$. We assume $T_i=(\alpha_i,\beta_i)^m$ where
$\alpha_i \in \mathbb{R} \cup \{-\infty\}$, $\beta_i \in
\mathbb{R} \cup \{\infty\}$, and $\alpha_i < \beta_i$.
We call such a type space an {\bf $m$-dimensional open interval domain}.
The set of all type profiles is denoted by $\mathbb{T}^n
= T_1 \times T_2 \times \ldots T_n$.

Let the set of all utility vectors for every alternative in $A$ be 
$\mathbb{D} \subseteq \mathbb{R}^n$, which is an open rectangle
in $\mathbb{R}^n$. Hence, the
set of type profiles can alternatively written as $\mathbb{D}^m$.
Throughout, we will require different mathematical properties of
$\mathbb{D}$ which are satisfied by $T_i$ for every $i$ if it is an $m$-dimensional
open interval domain. In particular, note the following two properties
which hold if the type space is
an $m$-dimensional open interval:

\begin{enumerate}

\item If we have a type profile $t$ in our domain
and permute two utility vectors $t^a$ and $t^b$ in this type profile, then
we will get a valid type profile in our domain.

\item For every type profile $t$ in our domain and every $a \in A$, there exists $\epsilon \gg 0$~
\footnote{For every pair of vectors $x,y \in \mathbb{R}^n$, we say that
$x \gg y$ if and only if $x$ is greater than $y$ in every component.}
such that if we increase the utility vector $t^a$ by $\epsilon$,
then we get a valid type profile in our domain.

\end{enumerate}

The first property follows from the interval assumption and the second
property follows from the openness assumption. 
We use these two properties extensively in our proofs.

We use the standard
notation of $t_{-i}$ to denote a type profile of agents
in $N \setminus \{i\}$ and $\mathbb{T}_{-i}$ to denote the
type spaces of agents in $N \setminus \{i\}$.

A social choice function is a mapping $f:\mathbb{T}^n \rightarrow A$.
A payment function is a mapping $p:\mathbb{T}^n \rightarrow \mathbb{R}^n$.
The payment of agent $i$ at type profile $t$ is denoted by $p_i(t)$.

\begin{defn}
A social choice function $f$ is implementable (in dominant
stragies) if there exists a payment function $p$ such that
for every $i \in N$ and every $t_{-i}$, we have
\begin{align*}
t_i^{f(t_i,t_{-i})} + p_i(t_i,t_{-i}) &\ge t_i^{f(s_i,t_{-i})} + p_i(s_i,t_{-i}) \qquad
\forall~s_i,t_i \in T_i.
\end{align*}
In this case, we say that $p$ implements $f$.
\end{defn}

Every social choice function satisfies certain properties if it is implementable. Below, we give one such useful property.

\begin{defn}
A social choice function $f$ satisfies {\bf positive association of differences (PAD)} if for every $s,t \in \mathbb{T}^n$ such that
$f(t)=a$ with $s^a-t^a \gg s^b-t^b$ for all $b \ne a$, we have $f(s)=a$.
\end{defn}
%
%\begin{defn}
%A social choice function $f$ satisfies {\bf negative association of differences (NAD)} if for every $s,t \in \mathbb{T}$ such that
%$f(t)=a$ with $s^a-t^a \gg s^b-t^b$ for some $b \ne a$, we have $f(s) \ne b$.
%\end{defn}

\begin{lemma}[\cite{Roberts79}]
Every implementable social choice function satisfies PAD.
\end{lemma}

A natural question to ask is what social choice functions are implementable. In an important result,
\cite{Roberts79} characterized the set of all social choice functions when the type space is
unrestricted and when the social choice function satisfies a condition called {\em non-imposition}.
\begin{defn}
A social choice function $f$ satisfies {\bf non-imposition} if for every $a \in A$, there exists $t \in \mathbb{T}^n$ such that $f(t)=a$.
\end{defn}

Using PAD and non-imposition, Roberts proved the following theorem.

\begin{theorem}[\cite{Roberts79}]\label{theo:roberts}
Suppose $T_i=\mathbb{R}^m$ for all $i \in N$. If $f$ is an implementable social choice function 
and satisfies non-imposition, then there exists weights $\lambda \in \mathbb{R}^n_+ \setminus \{0\}$ 
and a deterministic real-valued function $\kappa:A \rightarrow \mathbb{R}$ such that for all $t \in \mathbb{T}^n$,
\begin{align*}
f(t) &\in \arg \max_{a \in A}\big[\sum_{i \in N}\lambda_it_i^a-\kappa(a)\big]
\end{align*}
\end{theorem}

This family of social choice functions are called {\bf affine maximizer social choice functions}.

\subsection{Non-Affine-Maximizers in Bounded Domains: An Example}

Here, we give an example to illustrate that Theorem \ref{theo:roberts} does not hold in bounded domains.
The example is due to \cite{mtv03}.

\begin{example}\label{ex:ex1}
\end{example}
Let $N=\{1,2\}$ and $A=\{a,b,c\}$. Suppose $T_1=T_2=(0,1)^3$ (alternatively, suppose
$\mathbb{D}=(0,1)^2$). Consider the following allocation rule $f$. Let
\begin{align*}
\mathbb{T}^g &= \{(t_1,t_2) \in
\mathbb{T}^2: t_1^c < t_1^b + 0.5\} \cup \{(t_1,t_2) \in \mathbb{T}^2: t^c_2 > t^b_2 - 0.5\}.
\end{align*}
Then,
\begin{displaymath}
f(t_1,t_2) = \left \{
\begin{array}{l l}
\arg \max \{-1.5+t_1^a+t_2^a, t_1^b+t_2^b,t_1^c+t_2^c\} & \forall~(t_1,t_2) \in \mathbb{T}^g \\
c & \forall~(t_1,t_2) \in \mathbb{T}^2 \setminus \mathbb{T}^g.
\end{array} \right.
\end{displaymath}
It can be verified that $f$ satisfies non-imposition. Further, the following payment rule $p$ implements $f$.
\begin{displaymath}
p_1(t_1,t_2) = \left \{
\begin{array}{l l}
t^a_2 & \textrm{if}~f(t_1,t_2)=a \\
\min \{1.5+t^b_2,2+t^c_2\} & \textrm{if}~f(t_1,t_2)=b \\
(1.5+t^c_2) & \textrm{if}~f(t_1,t_2)=c.
\end{array} \right.
\end{displaymath}
\begin{displaymath}
p_2(t_1,t_2) = \left \{
\begin{array}{l l}
t^a_1 & \textrm{if}~f(t_1,t_2)=a \\
(1.5+t^b_1) & \textrm{if}~f(t_1,t_2)=b \\
\min \{1.5+t^c_1,2+t^b_1\} & \textrm{if}~f(t_1,t_2)=c.
\end{array}
\right.
\end{displaymath}
However, one can verify that $f$ is not an affine maximizer. 
In the example above, it is essential to assume that there are at
least two agents. Roberts' theorem holds in any bounded domain if there is only a single agent \cite{Chung06}.
\cite{Chung06} refer to this characterization as {\em pseudo-efficiency}.
%To see this,
%assume $N=\{1\}$. Let $f$ be an implementable social choice function and $p$ be a payment
%rule which implements $f$. Since the domain is bounded, $p$ is also a bounded function.
%The implementability condition then says
%\begin{align*}
%t_1^{f(t_1)} + p_1(t_1) &\ge t_1^{f(s_1)} + p_1(s_1) \qquad \forall~s_1,t_1 \in T_1.
%\end{align*}
%An equivalent way to write it in bounded domains is
%\begin{align*}
%t_1^{f(t_1)} + p_1(t_1) &= \max_{s_1 \in T_1}[t_1^{f(s_1)}+p_1(s_1)] \qquad \forall~t_1 \in T_1.
%\end{align*}
%It can be easily verified that if $f$ is implementable, then $p_1(s_1)=p_1(u_1)$ for all $s_1,u_1 \in T_1$
%with $f(s_1)=f(u_1)$. Hence, $p_1$ can be written as a mapping $p_1:A \rightarrow \mathbb{R}$ if $f$
%satisfies non-imposition.
%Hence, an equivalent way to write the implementability condition is
%\begin{align*}
%t_1^{f(t_1)} + p_1(f(t_1)) &= \max_{a \in A}[t_1^a + p_1(a)] \qquad \forall~t_1 \in T_1,
%\end{align*}
%which in turn is equivalent to writing
%\begin{align*}
%f(t_1) &\in \arg \max_{a \in A}[t_1^a + p_1(a)] \qquad \forall~t_1 \in T_1.
%\end{align*}
%This is the Roberts' theorem with $\kappa(a)=-p_1(a)$ for all $a \in A$ (note that in the one agent case,
%Roberts' theorem requires $\lambda_1=1$, which is the case here).

\subsection{Neutrality} \label{subsection:neutrality}

We restrict attention to neutral social choice functions which we now describe.
Neutrality roughly requires that
the mechanism designer should treat all allocations in $A$ symmetrically.
Given a social choice function $f$ we define the following set. For every $t \in \mathbb{T}^n$, the {\bf choice set} at $t$ is
defined as:
\begin{align*}
C^f(t) &= \{a \in A: ~\forall~\epsilon \gg 0~\textrm{and}~\forall~(t^a+\epsilon,t^{-a}) \in \mathbb{T}^n, f(t^a+\epsilon,t^{-a})=a\}.
\end{align*}

We first show that choice sets are non-empty under our assumptions of the domain ($m$-dimensional
open intervals).
\begin{lemma}\label{lem:new1}
Let $f$ be an implementable social choice function.
Then, for every $t \in \mathbb{T}^n$, $f(t) \in C^f(t)$.
\end{lemma}
\begin{proof}
Consider $t \in \mathbb{T}^n$, and let $f(t)=a$. Let $s=(s^a=t^a+\epsilon,s^{-a}=t^{-a})$ for some
$\epsilon \gg 0$ such that $s \in \mathbb{T}^n$. By PAD, $f(s)=a$. Hence, $a \in C^f(t)$.
\end{proof}

Using the notion of a choice set, we define a neutral social choice function. In Appendix B,
we discuss an alternate (but more standard) notion of neutrality, which we call {\em scf-neutrality},
defined directly on the social choice function, and show that scf-neutrality on implementable social choice functions
implies the following notion of neutrality.

\begin{defn}
A social choice function $f$ is {\bf neutral} if for every type profile $t \in \mathbb{T}^n$,
every permutation $\rho$ of $A$ and type profile $s$ induced by permutation $\rho$~\footnote{
Given a permutation $\rho$ of $A$ and a type profile $t$, the type profile induced
by permutation $\rho$ is the profile obtained from $t$ by relabeling the columns based
on permutation $\rho$.},
we have $C^f(s)=\{\rho(a): a \in C^f(t)\}$.
\end{defn}

A neutral social choice function does not discriminate
between social alternatives by their names. In many settings this is a natural
assumption. For instance, consider a city planner who has the following options
to improve public facilities in the city: (a) to build an opera house, (b) to build a public
school, (c) to build a park.
Although residents
of the city can have different (private) valuations over these alternatives, it is perfectly
reasonable to assume that the city planner has no preferences over
them.

Non-imposition is implied by neutrality in $m$-dimensional open interval domains.

\begin{lemma}\label{lem:ni}
Suppose $f$ is an implementable social choice function. If $f$ is neutral then it satisfies non-imposition.
\end{lemma}
\begin{proof}
Fix an alternative $a \in A$. Consider any arbitrary type
profile $t$ such that $f(t)=b \ne a$. By Lemma \ref{lem:new1},
$b \in C^f(t)$. Now, construct another type
profile $s=(s^a=t^b,s^b=t^a,s^{-ab}=t^{-ab})$. By neutrality,
$a \in C^f(s)$. Now, let $u=(u^a=s^a+\epsilon,u^{-a}=s^{-a})$
for some $\epsilon \in \mathbb{R}^n_{++}$. Since
$a \in C^f(s)$, we have that $f(u)=a$. Hence, $f$ satisfies non-imposition.
\end{proof}

Under neutrality, Roberts' theorem is modified straightforwardly as follows (see \cite{Lavi07}).
\begin{theorem}[\cite{Roberts79}]\label{theo:nroberts}
Suppose $T_i=\mathbb{R}^m$ for all $i \in N$. If $f$ is an implementable social choice function and satisfies neutrality,
then there exists weights $\lambda \in \mathbb{R}^n_+ \setminus \{0\}$ such that for all $t \in \mathbb{T}^n$,
\begin{align*}
f(t) &\in \arg \max_{a \in A}\sum_{i \in N}\lambda_it_i^a
\end{align*}
\end{theorem}

A striking aspect of this theorem is that it gives a precise functional form of the neutral social choice functions
that can be implemented. This family of social choice functions is called the {\bf weighted welfare maximizer social
choice functions}. If all the weights ($\lambda_i$s) are equal in a weighted welfare maximizer social
choice function, then we get the {\bf efficient social choice function}.

\section{An Induced Social Welfare Ordering}

In the aggregation theory literature, an axiom called ``binary independence"
is extensively used - see \cite{despremont02}. Roughly, it says that the comparision
between two alternatives $a$ and $b$ should only depend on the utility (column) vectors
corresponding to $a$ and $b$. We prove a counterpart of this axiom for our choice
set for $m$-dimensional open interval domains.

\begin{prop}[Binary Independence]\label{lem:new3}
Let $f$ be an implementable social choice function. 
Consider two type profiles $t=(t^a,t^b,t^{-ab}), s=(s^a=t^a,s^b=t^b,s^{-ab})$.
\begin{itemize}
\item[a)] Suppose $a,b \in C^f(t)$. Then, $a \in C^f(s)$ if and only if $b \in C^f(s)$.
\item[b)] Suppose $a \in C^f(t)$ but $b \notin C^f(t)$. Then $b \notin C^f(s)$.
\end{itemize}
\end{prop}
\begin{proof}
Suppose $a,b \in C^f(t)$. Now, consider a type profile $u=(u^a=t^a,u^b=t^b,u^{-ab})$, where $u_i^c=\min (t^c_i,s^c_i)$ for all $i \in M$ and for all $c \notin \{a,b\}$. Note that since $\mathbb{D}$ is an open
rectangle in $\mathbb{R}^n$, $u \in \mathbb{D}^m$.
\begin{itemize}

\item[a)] Suppose $a,b \in C^f(t)$. We will first show that $a,b \in C^f(u)$. Choose an $\epsilon \gg 0$. Since $a \in C^f(t)$, we know that $f(t^a+\frac{\epsilon}{2},t^b,t^{-ab})=a$. By PAD, $f(t^a+\epsilon,t^b,u^{-ab})=a$.
Hence, $a \in C^f(u)$. Using an analogous argument, $b \in C^f(u)$.

Now, suppose that $a \in C^f(s)$ and assume for contradiction $b \notin C^f(s)$.
Choose an $\epsilon \gg 0$ and arbitrarily close to zero. We show that
$f(t^a+2\epsilon,t^b+3\epsilon,s^{-ab}) \ne b$. Assume for contradiction,
$f(t^a+2\epsilon, t^b + 3\epsilon,s^{-ab})=b$. By PAD, $f(t^a,
t^b+4\epsilon,s^{-ab})=b$. Since $\epsilon$ can be made arbitrarily
small, this implies that $b \in C^f(s)$. This is a contradiction.

Next, we show that $f(t^a+2\epsilon,t^b+3\epsilon,s^{-ab}) \ne c$ for any
$c \notin \{a,b\}$. Assume for contradiction $f(t^a+2\epsilon,t^b+3\epsilon,s^{-ab})=c$
for some $c \notin \{a,b\}$. By PAD, $f(t^a+2\epsilon,t^b,s^c+\frac{\epsilon}{2},s^{-abc})=c$.
Also, since $a \in C^f(s)$, we know that $f(t^a+\epsilon,t^b,s^c,s^{-abc})=a$. By PAD,
$f(t^a+2\epsilon,t^b,s^c+\frac{\epsilon}{2},s^{-abc})=a$. This is a contradiction.

Hence, $f(t^a+2\epsilon,t^b+3\epsilon,s^{-ab})=a$. By PAD,
$f(t^a+\frac{5\epsilon}{2},t^b+3\epsilon,u^{-ab})=a$. We show that
$f(t^a,t^b+\epsilon',u^{-ab}) \ne b$ for all $0 \ll \epsilon' \ll \frac{\epsilon}{2}$.
Assume for contradiction $f(t^a,t^b+\epsilon',u^{-ab})=b$ for some
$0 \ll \epsilon' \ll \frac{\epsilon}{2}$. By PAD, $f(t^a+\frac{5\epsilon}{2},t^b+3\epsilon,u^{-ab})=b$.
This is a contradiction. Hence, $f(t^a,t^b+\epsilon',u^{-ab}) \ne b$ for some
$\epsilon' \gg 0$. This implies that $b \notin C^f(u)$, which is a contradiction.
Hence, $a \in C^f(s)$ implies that $b \in C^f(s)$.

Now, suppose that $a \notin C^f(s)$. Assume for contradiction
$b \in C^f(s)$. Exchanging the role of $a$ and $b$ above,
we get that $a \in C^f(s)$. This is a contradiction. Hence, if $a \notin C^f(s)$
then $b \notin C^f(s)$. This implies that either $\{a,b\} \subseteq C^f(s)$
or $\{a,b\} \cap C^f(s) = \emptyset$.

\item[b)] Suppose $a \in C^f(t)$ but $b \notin C^f(t)$. As in part (a),  $a \in C^f(u)$. Now, assume for contradiction,
$b \in C^f(s)$. If $a \notin C^f(s)$, then exchanging the role of $a$ and
$b$ in the second half of (a), we get that $a \notin C^f(u)$. This is a
contradiction. If $a \in C^f(s)$, then we have $a,b \in C^f(s)$ but
$a \in C^f(t)$. By part (a), $b \in C^f(t)$. This is a contradiction.

\end{itemize}
\end{proof}

We will define an ordering on $\mathbb{D}$ induced by an implementable 
social choice function. In general, we will refer to an arbitrary ordering $R$ 
on $\mathbb{D}$. The symmetric component of an ordering $R$ will be 
denoted as $I$ and the anti-symmetric component will be denoted as $P$.
Note that a social choice function $f$ is a mapping $f:\mathbb{T}^n \rightarrow A$. 
Hence, for every type profile $t$, a social choice function can be thought of as 
picking a column vector (which belongs to $\mathbb{D}$) in $t$. We will
show that in the process of picking these column vectors in $\mathbb{D}$ in 
an ``implementable manner", a neutral social choice function induces a social welfare 
ordering.

The following is a useful lemma that we will use in the proofs.
\begin{lemma}\label{lem:permute2}
Suppose $f$ is an implementable and neutral social choice function. 
Consider a type profile $t \in \mathbb{T}^n$ such that $t^a=t^b$ for
some $a,b \in A$. Then, $a \in C^f(t)$ if and only if $b \in C^f(t)$.
\end{lemma}
\begin{proof}
This follows from the fact that permuting columns $a$ and $b$ in $t$ produces $t$ again.
Hence, by neutrality, $a \in C^f(t)$ if and only if $b \in C^f(t)$.
\end{proof}

\begin{defn}
A {\bf social welfare ordering} $R^f$ induced by a social choice function
$f$ is a relation on $\mathbb{D}$ defined as follows.
The symmetric component of $R^f$ is denoted by $I^f$ and the
antisymmetric component of $R^f$ is denoted by $P^f$. Pick $x,y \in \mathbb{D}$.

We say $x P^f y$ if and only if there exists a profile $t$ with $t^a=x$ and $t^b=y$
for some $a,b \in A$ such that $a \in C^f(t)$ but $b \notin C^f(t)$.

We say $x I^f y$ if and only if there exists a profile $t$ with $t^a=x$ and $t^b=y$
for some $a,b \in A$ such that $a,b \in C^f(t)$.
\end{defn}

\begin{prop}[Social Welfare Ordering]\label{prop:order}
Suppose $f$ is an implementable and neutral social choice function.
Then, the relation $R^f$ induced by $f$ on $\mathbb{D}$ is an ordering.
\end{prop}
\begin{proof}
We first show that $R^f$ is well-defined. Pick $x,y \in \mathbb{D}$. We consider two cases. \\

\noindent {\sc Case 1:} Suppose $x P^f y$. Then there exists a type profile $t$ and some $a,b \in A$ with
$t^a=x$ and $t^b=y$ such that $a \in C^f(t)$ but $b \notin C^f(t)$. Consider any other type 
profile $s$ such that $s^a=x$ and $s^b=y$. By Proposition \ref{lem:new3}, $b \notin C^f(s)$. 
Consider any other profile $u$ and $(c,d) \ne (a,b)$ such that $u^c=x$ and $u^d=y$. 
We can permute $u$ to get another profile $v$ such that $v^a=x$ and $v^b=y$. By 
Proposition \ref{lem:new3}, $b \notin C^f(v)$. By neutrality, $d \notin C^f(u)$. 
Hence, the choice of $a$ and $b$ is without loss of generality,
i.e., for any $a,b \in A$ and any $t \in \mathbb{T}^n$ with $t^a=x$ and $t^b=y$, we 
have $b \notin C^f(t)$. So, $P^f$ is well-defined. \\

\noindent {\sc Case 2:} Suppose $x I^f y$. Then there exists a type profile $t$ and 
some $a,b \in A$ such that $a,b \in C^f(t)$. Consider any other type profile $s$ 
such that $s^a=x$ and $s^b=y$. By Proposition \ref{lem:new3}, $a \in C^f(s)$ if and only if 
$b \in C^f(s)$. By neutrality (as in Case 1), the choice of $a$ and $b$ is without loss 
of generality. This shows that $I^f$ is well-defined. \\

We next show that $R^f$ is reflexive.
Consider $x \in \mathbb{D}$ and the profile where $t^a=x$ for all $a \in A$.
By Lemma \ref{lem:permute2}, $C^f(t)=A$. Hence, $x I^f x$.

Next, we show that $R^f$ is complete. Choose $x,y \in \mathbb{D}$. Consider
a type profile $t$ where each column vector is either $x$ or $y$ with at least one column vector
being $x$ and at least one column vector being $y$. Suppose $f(t)=a$. 
Then either
$t^a=x$ or $t^a=y$. Without loss of generality, let $f(t)=a$ and $t^a=x$.
By Lemma \ref{lem:permute2}, there are two cases to consider. \\

\noindent {\sc Case 1:} For all $b$ with $t^b=y$ we have $b \in C^f(t)$.
Hence, $x I^f y$. \\

\noindent {\sc Case 2:} For all $b$ with $t^b=y$ we have $b \notin C^f(t)$.
Then, we get $x P^f y$. \\

This completes the argument that $R^f$ is complete, and hence, a binary relation.
Now, we prove that $R^f$ is transitive. Consider $x,y, z \in \mathbb{D}$.
Consider a type profile $t$, where each column has value in $\{x,y,z\}$ with at least
one column having value $x$, at least one column having value $y$, and at
least one column having value $z$ (this is possible since $|A|=m \ge 3$).

Due to Proposition \ref{lem:new3} and neutrality, 
without loss of generality let $t^a=x,t^b=y,t^c=z$.
We prove transitivity of $P^f$ and $I^f$, and this implies transitivity of $R^f$.
\\

\noindent {\sc Transitivity of $P^f$:} Suppose $x P^f y$ and $y P^f z$. This implies
that $a \in C^f(t)$ but $b \notin C^f(t)$.
Since $y P^f z$, we get that $c \notin C^f(t)$. Since $c \notin C^f(t)$, we have $x P^f z$. \\

\noindent {\sc Transitivity of $I^f$:} Suppose $x I^f y$ and $y I^f z$. This implies that
$a,b \in C^f(t)$. But $y I^f z$ implies that $c \in C^f(t)$ too. This implies that $x I^f z$.

\end{proof}

\section{Properties of the Induced Social Welfare Ordering}

In this section, we fix an implementable neutral social choice function $f$. We then 
prove that the social welfare ordering $R^f$ defined in the last section satisfies three 
specific properties.

\begin{defn}
An ordering $R$ on $\mathbb{D}$ satisfies {\bf weak Pareto} if  for all
$x,y \in \mathbb{D}$ with $x \gg y$ we have $x P y$.
\end{defn}

\begin{defn}
An ordering $R$ on $\mathbb{D}$ satisfies {\bf invariance} if for all
$x,y \in \mathbb{D}$ and all $z \in \mathbb{R}^n$ such that
$(x+z),(y+z) \in \mathbb{D}$ we have $x P y$ implies $(x+z)  P (y+z)$
and $x I y$ implies $(x+z) I (y+z)$.
\end{defn}

\begin{defn}
An ordering $R$ on $\mathbb{D}$ satisfies {\bf continuity} if for all
$x \in \mathbb{D}$, the sets $U^x=\{y \in \mathbb{D}:y R x\}$
and $L^x=\{y \in \mathbb{D}:x R y\}$ are closed in $\mathbb{D}$.
\end{defn}

\begin{prop}[Axioms for Social Welfare Ordering]\label{prop:axioms}
Suppose $f$ is an implementable and neutral social choice function.
Then the social welfare ordering $R^f$ induced by $f$ on $\mathbb{D}$
satisfies weak Pareto, invariance, and continuity.
\end{prop}
\begin{proof}
We show that $R^f$ satisfies each of the properties. \\

\noindent {\sc Weak Pareto:} Choose $x,y \in \mathbb{D}$ such that $x \gg y$.
Start with a profile $t$ where $t^a=y$ for all $a \in A$. Suppose $f(t)=b$. Consider 
another profile $s=(s^b=x,s^{-b}=t^{-b})$ (i.e. column vector corresponding to $b$ 
is changed from $y$ to $x$). By PAD, $f(s)=b$ and hence $b \in C^f(s)$. We show 
that for any $a \ne b$ we have $a \notin C^f(s)$. Choose $\epsilon \gg 0$ but 
$\epsilon \ll x-y$. By PAD, $f(t^a+\epsilon,s^b=x,t^{-ab})=b$. Hence, $a \notin C^f(s)$. 
This shows that $b \in C^f(s)$ but $a \notin C^f(s)$. Hence, by Proposition 
\ref{prop:order}, $x P^f y$. \\

\noindent {\sc Invariance:} Choose $x,y \in \mathbb{D}$ and $z \in \mathbb{R}^n$ 
such that $(x+z),(y+z) \in \mathbb{D}$. We consider two cases. \\

\noindent {\sc Case 1:} Suppose $x P^f y$. We show that $(x+z) P^f (y+z)$. 
Since $x P^f y$, there exists a profile $t=(t^a=x,t^b=y,t^{-ab})$ such that $a \in C^f(t)$ 
but $b \notin C^f(t)$. Consider the profile $s$, where $s^c=t^c+z$ for all $c \in A$. 
Fix $\epsilon \gg 0$. Since $a \in C^f(t)$, $f(t^a+\frac{\epsilon}{2},t^b,t^{-ab})=a$. 
Hence, by PAD $f(s^a+\epsilon,s^b,s^{-ab})=a$. This shows that $a \in C^f(s)$. 
Since $b \notin C^f(t)$, there is some $\epsilon \gg 0$ such that 
$f(t^a,t^b+\epsilon,t^{-ab}) \ne b$. We show that 
$f(s^a,s^b+\frac{\epsilon}{2},s^{-ab}) \ne b$. Assume for 
contradiction $f(s^a,s^b+\frac{\epsilon}{2},s^{-ab})=b$. By PAD, 
$f(t^a,t^b+\epsilon,t^{-ab})=b$. This is a contradiction. Hence, 
$f(s^a,s^b+\frac{\epsilon}{2},s^{-ab}) \ne b$. This implies that 
$b \notin C^f(s)$. Using Proposition \ref{prop:order}, $(x+z) P^f (y+z)$. \\

\noindent {\sc Case 2:} Suppose $x I^f y$. We show that $(x+z) I^f (y+z)$. 
Then, there exists a profile $t=(t^a=x,t^b=y,t^{-ab})$ such that $a,b \in C^f(t)$. 
Consider the profile $s$, where $s^c=t^c+z$ for all $c \in A$. Fix $\epsilon \gg 0$. 
Since $a \in C^f(t)$, $f(t^a+\frac{\epsilon}{2},t^b,t^{-ab})=a$. Hence, 
by PAD $f(s^a+\epsilon,s^b,s^{-ab})=a$. This shows that $a \in C^f(s)$. 
Using an analogous argument, $b \in C^f(s)$. Hence, by Proposition 
\ref{prop:order}, $(x+z) I^f (y+z)$. \\

\noindent {\sc Continuity:} Fix $x \in \mathbb{D}$. We show that the 
set $U^x=\{y \in \mathbb{D}:y R^f x\}$ is closed. Take an infinite 
sequence $y_1,y_2,\ldots$ such that every point $y_n$ in this sequence 
satisfies $y_n R^f x$. Let this sequence converge to $z \in \mathbb{D}$. 
Assume for contradiction $x P^f z$. Consider a type profile $t$ such that 
$t^a=x$ and $t^c=z$ for all $c \ne a$. Since $x P^f z$, we have 
$c \notin C^f(t)$ for all $c \ne a$. Hence, $C^f(t)=\{a\}$.

Consider $b \ne a$. Since $b \notin C^f(t)$, we know that there 
exists $\epsilon \gg 0$ and $\epsilon$ arbitrarily close to the zero 
vector such that $f(t^a,t^b+\epsilon,t^{-ab}) \ne b$. We show that 
$f(t^a,t^b+\epsilon,t^{-ab}) \ne c$ for all $c \notin \{a,b\}$. Assume 
for contradiction $f(t^a,t^b+\epsilon,t^c,t^{-abc})=c$ for some 
$c \notin \{a,b\}$. Then, by PAD, $f(t^a,t^b,t^c+\epsilon'',t^{-abc})=c$ 
for all $\epsilon'' \gg 0$. This implies that $c \in C^f(t)$, which is a 
contradiction. Hence, $f(t^a,t^b+\epsilon,t^{-ab})=a$.

This implies that $xR^f(z+\epsilon)$.
Since the sequence converges to $z$, there is a point $z' \in \mathbb{D}$ 
arbitrarily close to $z$ such that $z' R^f x$. Since $z$ is arbitrarily close 
to $z'$, by weak Pareto, $(z+\epsilon) P^f z'$. Using $z' R^f x$, we 
get $(z+\epsilon) P^f x$. This is a contradiction to the fact that $x R^f (z+\epsilon)$.

To show $L^x=\{y \in \mathbb{D}:x R^f y\}$ is closed, take an infinite 
sequence $y_1,y_2,\ldots$ such that every point $y_n$ in this sequence 
satisfies $x R^f y_n$. Let this sequence converge to $z$. Assume for 
contradiction $z P^f x$. Interchanging the role of $x$ and $z$ in the 
previous argument, we will get $z R^f (x+\epsilon)$ for some $\epsilon \gg 0$. 
Since the sequence converges to $z$, there is a point $z' \in \mathbb{D}$ 
arbitrarily close to $z$ such that $x R^f z'$. Since $z'$ is arbitrarily close to 
$z$, $(x+\epsilon) P^f z$ by weak Pareto. This is a contradiction to the 
fact that $z R^f (x+\epsilon)$.
\end{proof}

\section{Multi-dimensional Open Interval Domains}

In this section, we prove the main result.
In particular, we prove a proposition related to 
linear utility representation on open and convex sets.

\begin{prop}[Representation of Social Welfare Ordering]
\label{prop:reporder}
Suppose an ordering $R$ on $\mathbb{D}$ satisfies weak Pareto, invariance,
and continuity. If $\mathbb{D}$ is open and convex, then there exists weights
$\lambda \in \mathbb{R}^n_+ \setminus \{0\}$ and for all $x,y \in \mathbb{D}$
\begin{align*}
x R y &\Leftrightarrow \sum_{i \in N}\lambda_ix_i \ge \sum_{i \in N}\lambda_iy_i.
\end{align*}
\end{prop}
\begin{proof}
Fix any $z \in \mathbb{D}$. Denote $U^z=\{x:x R z\}$,
$L^z=\{x: z R x\}$, $\mathbb{D} \setminus L^z = \{x: x P z\}$,
and $\mathbb{D} \setminus U^z=\{x: z P x\}$. \\

\noindent {\sc Step 1:} We first show that the sets $U^z,L^z,\mathbb{D}
\setminus U^z$, and $\mathbb{D} \setminus L^z$ are convex.
We make use of the following fact here.

\begin{fact}\label{fact:f1}
Consider a set $X \subseteq \mathbb{D}$
and let $X$ satisfy the property that if $x,y \in X$ then $\frac{x+y}{2} \in X$.
If $X$ is open in $\mathbb{D}$ or closed in $\mathbb{D}$, then
$X$ is convex.
\end{fact}
The proof of this fact is given in the Appendix A. By continuity, each of the
sets $U^z, L^z, \mathbb{D} \setminus U^z$, and $\mathbb{D} \setminus
L^z$ are either open or closed in $\mathbb{D}$. Hence, by Fact
\ref{fact:f1}, we only need to verify that these sets are closed under the midpoint
operation.

Consider $U^z$. Now, let $x,y \in \mathbb{D}$ such that
$x R z$ and $y R z$. We will show that $\frac{x+y}{2} R z$.
Note that $\frac{x+y}{2} \in \mathbb{D}$ because
$\mathbb{D}$ is convex. Now, assume for contradiction that $z P \frac{x+y}{2}$. This implies that
$x P \frac{x+y}{2}$ and $y P \frac{x+y}{2}$. By invariance,
$x + \frac{y-x}{2} P \frac{x+y}{2} + \frac{y-x}{2}$. Hence,
$\frac{x+y}{2} P y$. This is a contradiction. Hence,
the set $U^z$ is convex.

Similar arguments show that
$L^z$, $\mathbb{D} \setminus L^z$, and
$\mathbb{D} \setminus U^z$ are convex. \\

\noindent {\sc Step 2:} We now show that $z$ is a boundary point
of $U^z$. Let $B_{\delta}(z)=\{x: \Vert x-z \Vert < \delta\}$, where $\delta \in \mathbb{R}_+$. Since $\mathbb{D}$
is open, there exists $\epsilon \gg 0$ such that $(z+\epsilon) \in
\mathbb{D} \cap B_{\delta}(z)$ and, by weak Pareto, $(z+\epsilon) P z$. Further,
since $\mathbb{D}$ is open, $\epsilon$ can be chosen such that
$(z-\epsilon) \in \mathbb{D} \cap B_{\delta}(z)$, and
by weak Pareto, $z P (z-\epsilon)$. Hence, for every $\delta > 0$, there exists
a point in $B_{\delta}(z)$ which is in $U^z$ and another point which is not in $U^z$.
This shows that $z$ is a boundary point of $U^z$. \\

\noindent {\sc Step 3:} By the supporting hyperplane theorem,
there exists a hyperplane through $z$ supporting the set $U^z$, i.e.,
there exists a non-zero vector $\lambda \in \mathbb{R}^n \setminus \{0\}$ such that for all $x \in U^z$,
\begin{align*}
\sum_{i=1}^n\lambda_ix_i &\ge \sum_{i=1}^n\lambda_iz_i.
\end{align*}
Denote the intersection of this hyperplane with the set $\mathbb{D}$
as $H^z$. \\

\noindent {\sc Step 4:} We next show that $\lambda \in \mathbb{R}^n_+ \setminus \{0\}$.
Assume for contradiction
$\lambda_j < 0$ for some $j \in N$. Since $\mathbb{D}$ is open
there exists $\epsilon \gg 0$ such that $(z+\epsilon) \in \mathbb{D}$.
Moreover, we can choose $\epsilon$ such that
\begin{align*}
\sum_{i=1}^n\lambda_i \epsilon_i &< 0.
\end{align*}
By weak Pareto $(z+\epsilon) P z$. Hence, $(z+\epsilon) \in U^z$. Thus,
\begin{align*}
\sum_{i=1}^n\lambda_i(z_i+\epsilon_i) &\ge \sum_{i=1}^n\lambda_iz_i.
\end{align*}
This implies that
\begin{align*}
\sum_{i=1}^n\lambda_i \epsilon_i &\ge 0.
\end{align*}
This is a contradiction. Hence, $\lambda_i \ge 0$ for all $i \in N$. \\

\noindent {\sc Step 5:} Now, consider $x \in \mathbb{D}$ such that
\begin{align*}
\sum_{i=1}^n\lambda_ix_i &> \sum_{i=1}^n\lambda_iz_i.
\end{align*}
We will show that $x P z$. Assume for contradiction $z R x$. We consider two cases. \\

\noindent {\sc Case 1:} Suppose $z P x$. Since $\mathbb{D}$ is open, there exists
a point $z'$ in $B_{\delta}(z)$ for some $\delta \in \mathbb{R}_+$ such that
\begin{enumerate}
\item[a)] $z$ lies on the line segment joining $z'$ and $x$ and
\item[b)] $x$ and $z'$ lies on opposite sides of the hyperplane $H_z$, i.e.,
\begin{align*}
\sum_{i=1}^n\lambda_iz'_i &< \sum_{i=1}^n\lambda_iz_i.
\end{align*}
\end{enumerate}

By (b) and using Step 3, $z P z'$. By our assumption $z P x$. Hence, $x, z' \in \mathbb{D}
\setminus U^z$. By Step 1, $\mathbb{D} \setminus U^z$ is convex.
Since $z$ is in the convex hull of $x$ and $z'$, we get that $z P z$. This is a contradiction. \\

\noindent {\sc Case 2:} Suppose $z I x$. Since $\mathbb{D}$ is open,
there exists $x'=x-\epsilon$ for some $\epsilon \gg 0$ such that
\begin{align*}
\sum_{i=1}^n\lambda_ix'_i &> \sum_{i=1}^n\lambda_iz_i.
\end{align*}
By weak Pareto $x P x'$. Hence, $z P x'$. By Case 1, this is not possible.
This is a contradiction. \\

\noindent Hence, in both cases we reach a contradiction, and conclude that $x P z$. \\

\noindent {\sc Step 6:} Now, consider $x \in \mathbb{D}$ such that
\begin{align*}
\sum_{i=1}^n\lambda_ix_i &= \sum_{i=1}^n\lambda_iz_i.
\end{align*}
We will show that $x I z$. Suppose not. There are two cases to consider. \\

\noindent {\sc Case 1:} Assume for contradiction $x P z$. By continuity, the
set $\{y: y P z\}$ is open in $\mathbb{D}$. Since $\mathbb{D}$ is open
in $\mathbb{R}^n$, we get that $\{y: y P z\}$ is open in $\mathbb{R}^n$.
Hence, there exists $\delta \in \mathbb{R}_+$
such that for every point in $x' \in B_{\delta}(x)$ we have $x' P z$. Choose
$\epsilon \gg 0$ such that for $x''=x-\epsilon$ we have $x'' \in B_{\delta}(x)$.
Hence, $x'' P z$. By Step 4, $\lambda \in \mathbb{R}^n_+ \setminus \{0\}$.
Hence, we get
\begin{align*}
\sum_{i=1}^n\lambda_ix''_i &< \sum_{i=1}^n\lambda_iz_i.
\end{align*}
But this is a contradiction since $x'' P z$ implies $x'' \in U^z$, which in turn implies
that
\begin{align*}
\sum_{i=1}^n\lambda_ix''_i &\ge \sum_{i=1}^n\lambda_iz_i.
\end{align*}

\noindent {\sc Case 2:} Assume for contradiction $z P x$. By continuity, the
set $\{y: z P y\}$ is open in $\mathbb{D}$. Hence, there exists $\delta \in \mathbb{R}_+$
such that for every point in $x' \in B_{\delta}(x)$ we have $z P x'$. Choose
$\epsilon \gg 0$ such that for $x''=x+\epsilon$ we have $x'' \in B_{\delta}(x)$.
Hence, $z P x''$. By Step 4, $\lambda \in \mathbb{R}^n_+ \setminus \{0\}$.
Hence, we get
\begin{align*}
\sum_{i=1}^n\lambda_ix''_i &> \sum_{i=1}^n\lambda_iz_i.
\end{align*}
By Step 5, this implies that $x'' P z$. This is a contradiction.

This shows that for any $z$, there exists $\lambda \in \mathbb{R}^n_+ \setminus \{0\}$
such that for all $x \in \mathbb{D}$, we have
\begin{align*}
x R z &\Leftrightarrow \sum_{i=1}^n\lambda_ix_i \ge \sum_{i=1}^n\lambda_iz_i.
\end{align*}

In other words, $H^z$ contains all the points in $\mathbb{D}$ which
are indifferent to $z$ under $R$. Moreover, on one side of $H^z$ we have points
in $\mathbb{D}$ which are better than $z$ under $R$ and on the other side,
we have points which are worse than $z$ under $R$.

Finally, pick any two points $x$ and $y$ in $\mathbb{D}$. 
Since $\mathbb{D}$ is open and convex, we can connect
$x$ and $y$ by a series of intersecting open balls along the convex hull
of $x$ and $y$, with each of these open balls contained in $\mathbb{D}$.
By invariance, for any two points $x'$ and $y'$ in such an open
ball, $H^{x'}$ and $H^{y'}$ have to be parallel to each other.
Since such open balls intersect each other, the hyperplanes $H^x$ and $H^y$ are parallel to each other. This 
completes the proof.
\end{proof}

When $\mathbb{D}= \mathbb{R}^n$, this result is well known due to \cite{Blackwell54} (see also
recent proofs in the utility representation literature - \cite{despremont77},
\cite{Blackorby84}, \cite{Trockel92},
and \cite{despremont02}).

We are now ready to state our main result.
\begin{theorem}\label{theo:newroberts}
Suppose $f$ is a neutral social choice function and for every $i \in N$, $T_i$ is an 
$m$-dimensional open interval. 
The social choice function $f$ is implementable if and only if there exists weights
$\lambda \in \mathbb{R}^n_+ \setminus \{0\}$ such that for all $t \in \mathbb{T}^n$,
\begin{align*}
f(t) &\in \arg \max_{a \in A} \sum_{i \in N}\lambda_it^a_i.
\end{align*}
\end{theorem}
\begin{proof}
Suppose $f$ is neutral and implementable.
Note that since for every $i \in N$, $T_i$ is an open interval domain, then
$\mathbb{D}$ must be convex and open in $\mathbb{R}^n$ - indeed,
$\mathbb{D}$ is an open rectangle in $\mathbb{R}^n$. 
Hence, by Proposition \ref{prop:order}, a neutral and implementable
SCF $f$ induces a social welfare ordering $R^f$ on $\mathbb{D}$.
By Proposition \ref{prop:axioms}, $R^f$ satisfies continuity, weak Pareto,
and invariance. By Proposition \ref{prop:reporder} (since $\mathbb{D}$
is open and convex), there
exists weights $\lambda \in \mathbb{R}^n_+ \setminus \{0\}$
such that for every $x,y \in \mathbb{D}$
we have
\begin{align*}
x R^f y &\Leftrightarrow \sum_{i \in N}\lambda_ix_i \ge \sum_{i \in N}\lambda_iy_i.
\end{align*}
Finally, by Lemma \ref{lem:new1} for all $t \in \mathbb{D}^m$, $f(t) \in C^f(t)$. Hence,
$t^{f(t)} R^f t^b$ for all $b \in A$ and for all $t \in \mathbb{D}^m$.

It is well known that if $f$ is a weighted welfare
maximizer with weights $\lambda \in \mathbb{R}^n_+ \setminus \{0\}$, then
the following payment function $p:\mathbb{T}^n \rightarrow \mathbb{R}^n$
makes the social choice function implementable. For all $i \in N$ with $\lambda_i=0$,
$p_i(t)=0$ for all $t \in \mathbb{T}^n$. For all $i \in N$ with $\lambda_i > 0$,
\begin{align*}
p_i(t) &= \frac{1}{\lambda_i}\big[ \sum_{j \ne i}\lambda_jt_j^{f(t)} \big] - h_i(t_{-i}) \qquad
\forall~t \in \mathbb{T}^n.
\end{align*}
where $h_i: \mathbb{T}_{-i} \rightarrow \mathbb{R}$~\footnote{
Since $T_i$ is
connected for all $i \in N$, revenue equivalence holds in this setting \cite{Chung07,Hyden09}.
Hence, these are the {\em only} payment functions which makes $f$ implementable.}.
This proves the theorem.
\end{proof}

\subsection{Discussions}

In this section, we make several observations relating to our results. \\

\noindent{\sc Affine Maximizer and Weighted Welfare Maximizer Domains.} A plausible
conjecture is that every domain where neutral and implementable social choice functions
are weighted welfare maximizers are also domains where implementable social choice functions
are affine maximizers. This conjecture is false. To see this, observe that the domain in Example \ref{ex:ex1}.
The domain in this example, $(0,1)^2$ is a 2-dimensional open interval domain. However we have already
seen that it admits implementable social choice functions that are non-affine-maximizers (of course, 
these social choice functions are not neutral). 
This observation emphasizes the fact that neutrality plays a critical role in our result.\\

\noindent {\sc Auction domains are not covered.} It is well known that in auction domains,
there are social choice functions other than affine maximizers which are implementable
\cite{Lavi03}. These social choice functions are also neutral. Hence, in auction
domains, there are neutral social choice functions which are implementable, but not
weighted welfare maximizers. It can be reconciled with our result in several ways.
First, auction domains
are restricted domains which are not necessarily open (or even full dimensional). For example, consider the sale of
two objects to two buyers. The set of allocations can be $\{a,b,c,d\}$, where
$a$ denotes buyer 1 gets both the objects, $b$ denotes buyer 2 gets
both the objects, $c$ denotes buyer 1 gets object 1
and buyer 2 gets object 2, and $d$ denotes buyer 1 gets object 2 and buyer 2 gets
object 1. Note here that in every utility vector $t^a$ for allocation $a$ buyer 2 will have
zero valuation. Similarly, in every utility vector $t^b$ for allocation $b$ buyer 1 will have
zero valuation. Hence, this domain is not open.

Second,
our open interval domain assumption
is not usually satisfied in auction domains. This is because, agents usually have a partial
order on the set of alternatives (see \cite{Bikh06}). We do not allow any such partial order
in our model. Finally, neutrality is an unacceptably restrictive assumption in auction domains. 

However, as we have noted in Section \ref{subsection:neutrality}, there are settings where our domain
and neutrality assumptions are plausible. \\

\noindent {\sc No ordering without neutrality.} If we drop neutrality and replace it with
non-imposition, then Roberts' theorem says that
affine maximizers (as in Theorem \ref{theo:roberts}) are the only implementable social
choice functions. But affine maximizers do not necessarily induce the ordering we discussed.
This is because of the $\kappa(\cdot)$ terms in the affine maximizers. For example,
consider a type profile $t=(t^a=x,t^b=y,t^{-ab})$. Suppose $a \in C^f(t)$ but $b \notin C^f(t)$.
Here, the $\kappa(a)$ term may be higher than $\kappa(b)$ such that when we permute
the columns of $a$ and $b$ and get the new type profile $s=(s^a=y,s^b=x,t^{-ab})$, we
still have $a \in C^f(s)$ and $b \notin C^f(s)$. Thus, our social welfare ordering is
not induced here. \\

\noindent {\sc Anonymity gives efficiency.} Consider the following additional condition on every social choice function.

\begin{defn}
A social choice function $f$ is {\bf anonymous} if for every $t \in \mathbb{T}^n$ and
every permutation $\sigma$ on the row vectors (agents) of $t$, we have $f(\sigma(t))=f(t)$.
\end{defn}

\begin{defn}
An ordering $R$ on $\mathbb{D}$ satisfies {\bf anonymity} if for every $x,y \in \mathbb{D}$
and every permutation $\sigma$ on agents we have $x I y$ if $x=\sigma(y)$.
\end{defn}

\begin{lemma}\label{lem:anon}
Suppose $f$ is implementable and anonymous. Then, $R^f$ satisfies anonymity.
\end{lemma}
\begin{proof}
Let $\sigma$ be a permutation of the set of agents.
For any vector $x \in \mathbb{D}$, we write $\bar{\sigma}(x)$ to denote
the permutation of vector $x$ induced by the permutation $\sigma$ on
sets of agents.
Consider $x,y \in \mathbb{D}$ such that $y=\bar{\sigma}(x)$. 
Assume for contradiction $x P^f y$. Consider a type profile
$t$ such that $t^a=x$ and $t^b=y$ for all $b \ne a$. Hence, $C^f(t)=\{a\}=f(t)$.
Let $s$ be the type profile such that $s^c=\bar{\sigma}(t^c)$ for all $c \in A$.
Since $f$ is anonymous $f(s)=a$. Hence, $y R^f \bar{\sigma}(y)$, which futher
implies that $x P^f \bar{\sigma}(\bar{\sigma}(x))$.
Repeating this argument again, we will get $\bar{\sigma}(y) R^f \bar{\sigma}(\bar{\sigma}(y))$.
Hence, $x P^f \bar{\sigma}(\bar{\sigma}(\bar{\sigma}(x)))$.
Clearly, after repeating this procedure some finite number of times, 
we will be able to conclude $x P^f x$, which is a contradiction.
\end{proof}

It is straightforward to show using Theorem \ref{theo:newroberts} 
that every implementable, neutral,
and anonymous social choice function in an open interval domain is the efficient social choice function.
Here, we show that this result holds for some other domains too.
The proof is an adaptation of an elegant proof
by \cite{Milnor54} (see also Theorem 4.4 in \cite{despremont02}). We give the proof in Appendix A.

\begin{theorem}
\label{theo:efficiency} Suppose $f$ is implementable, neutral, and anonymous. If
$\mathbb{T}^n=[0,H)^{m \times n}$, where $H \in \mathbb{R} \cup \{\infty\}$, then $f$ is 
the efficient social choice function.
\end{theorem}

Note here that the domain in Theorem \ref{theo:efficiency} always includes
the origin (this is crucial for the proof) and is not open from ``left". Hence, this result is not a corollary to
Theorem \ref{theo:newroberts}.

\section{Roberts' Affine-Maximizer Theorem}
\label{sec:gen}

In this section,
we show how the general version of Roberts' theorem using version of 
Roberts' theorem with neutrality, which we have proved earlier.
We assume throughout that the domain is {\bf unrestricted}, i.e., $\mathbb{T}^n=\mathbb{R}^{m \times n}$.
Although our proof of the general Roberts' theorem uses elements
developed in earlier proofs, we believe nonetheless that it offers some
new insights into the result. The main idea behind our proof is to transform an 
arbitrary implementable
social choice function to a neutral implementable social choice function. Then, we can readily use Roberts'
theorem with neutrality on the new social choice function to get the Roberts' theorem.

Consider a mapping $\delta: A \rightarrow \mathbb{R}$.
Denote $1_{\delta(a)}$ as the vector of $\delta(a)$ s in $\mathbb{R}^n$.
Let $1_{\delta} \equiv (1_{\delta(a)},1_{\delta(b)},\ldots)$ be
the profile of $m$ such vectors, each corresponding to an allocation in $A$.
For any social choice function $f$, define $f^{\delta}$ as follows. For every $t \in \mathbb{T}^n$,
let $(t+ 1_{\delta}) \in \mathbb{T}^n$ be such that $(t+1_{\delta})^a=t^a + 1_{\delta(a)}$ for all $a \in A$.
For every $t \in \mathbb{T}^n$, let
\begin{align*}
f^{\delta}(t) &= f(t+1_{\delta}).
\end{align*}
Since $\delta(a)$ is finite for all $a \in A$, the social choice function $f^{\delta}$
is well-defined.

\begin{prop}[Implementability Invariance]\label{prop:unr1}
For every $\delta: A \rightarrow \mathbb{R}$, if $f$ is implementable, then $f^{\delta}$ is implementable.
\end{prop}
\begin{proof}
Since $f$ is implementable, there exists a payment function $p$ which implements it.
We define another payment function $p^{\delta}$ as follows. For every $t \in \mathbb{T}^n$
and every $i \in N$,
\begin{align*}
p_i^{\delta}(t) &= p_i(t+1_{\delta}) + \delta(f^{\delta}(t)).
\end{align*}
We will show that $p^{\delta}$ implements $f^{\delta}$. To see this, fix an agent $i \in N$
and $t_{-i} \in \mathbb{T}_{-i}$. Let $s=(s_i,t_{-i})$ and note the following.
\begin{align*}
t_i^{f^{\delta}(t)} + p^{\delta}_i(t) &= t^{f(t+1_{\delta})}_i + p_i(t+1_{\delta}) + \delta(f^{\delta}(t)) \\
&= t^{f(t+1_{\delta})}_i + p_i(t+1_{\delta}) + \delta(f(t+1_{\delta})) \\
&= (t+1_{\delta})^{f(t+1_{\delta})}_i + p_i(t+1_{\delta}) \\
&\ge (t+1_{\delta})^{f(s+1_{\delta})}_i + p_i(s+1_{\delta}) \\
&= (t+1_{\delta})^{f^{\delta}(s)}_i + p_i(s+1_{\delta}) \\
&= t_i^{f^{\delta}(s)} + \delta(f^{\delta}(s)) + p_i(s+1_{\delta}) \\
&= t_i^{f^{\delta}(s)} + p^{\delta}(s), 
\end{align*}
where the inequality followed from the implementability of $f$ by $p$. Hence, $p^{\delta}$
implements $f^{\delta}$.
\end{proof}

Our next step is to find a mapping $\delta:A \rightarrow \mathbb{R}$ such that
$f^{\delta}$ is neutral. We will need
the following property of choice sets.

\begin{lemma}\label{lem:ch1}
Suppose $f$ is implementable and satisfies non-imposition. Let $t$ be
a type profile such that $C^f(t)=\{a\}$ for some $a \in A$. Then,
for some $\epsilon \in \mathbb{R}^n_{++}$, $a \in C^f(s)$, where
$s^a=t^a-\epsilon$ and $s^b=t^b$ for all $b \ne a$.
\end{lemma}
\begin{proof}
Since $C^f(t)=\{a\}$, we have $f(t)=a$ (by Lemma \ref{lem:new1}).
Choose some $b \ne a$. Since $b \notin C^f(t)$, there exists
$\epsilon_b \in \mathbb{R}^n_{++}$ such that $b \notin C^f(u)$,
where $u^b=t^b+\epsilon_b$ and $u^c=t^c$ for all $c \ne b$.
Indeed, by Proposition \ref{lem:new3}, $C^f(u)=\{a\}$.
Now, consider the type profile $v$ such that $v^c=t^c+\epsilon_c$
for all $c \ne a$ and $v^a=t^a$. We will show that $C^f(v)=\{a\}$.

To show this, we go from $t$ to $v$ in $(m-1)$ steps. In the first step,
we choose an arbitrary allocation $b \ne a$, and consider a type profile
$x$, where $x^b=v^b$ and $x^c=t^c$ for all $c \ne b$. By definition of
$\epsilon_b$, we have $C^f(x)=\{a\}$. Next, we choose another allocation
$c \notin \{a,b\}$, and consider a type profile $y$ such that $y^d=x^d$ if $d \ne c$
and $y^d=v^d$ otherwise. We first show that $c \notin C^f(y)$. Assume for contradiction,
$c \in C^f(y)$, then by PAD, $c \in C^f(x)$. This is a contradiction. Hence, $c \notin C^f(y)$,
and by Proposition \ref{lem:new3}, $C^f(y)=\{a\}$. We now repeat this procedure
by choosing $d \notin \{a,b,c\}$ and considering a type profile $z$ where utility vector
of $z$ is increased to $v^d$ and every other utility vector remains at $y$. After
a finite steps, we will reach the type profile $v$ with $C^f(v)=\{a\}$.

Now, choose $\epsilon=\frac{1}{2} \min_{b \ne a}\epsilon_b$. Consider a type profile
$s$ such that $s^b=t^b$ for all $b \ne a$ and $s^a=t^a-\epsilon$. By PAD (from
$v$ to $s$), $a \in C^f(s)$.
\end{proof}

Now, we define a set which can also be found in Roberts' original proof (see also
\cite{Lavi08}). For every $a, b \in A$
and every social choice function $f$ define the $P$-set as
\begin{align*}
P^f(a,b) &= \{\alpha \in \mathbb{R}^n: \exists t \in \mathbb{T}^n~\textrm{such that}~
a \in C^f(t), t^a-t^b=\alpha\}.
\end{align*}

\cite{Roberts79} and \cite{Lavi08} define the $P$-set slightly differently. They
let $P^f(a,b)=\{\alpha \in \mathbb{R}^n: \exists t \in \mathbb{T}^n~\textrm{such that}~
f(t)=a, t^a-t^b=\alpha\}$. Our notion of $P$-set is the interior of the $P$-set
they define.

The $P$-sets are non-empty if the social choice function
satisfies non-imposition. To see this, choose $a,b \in A$ and a social choice function $f$. 
By non-imposition, there
must exist a $t \in \mathbb{T}^n$ such that $f(t)=a$, which implies
that $a \in C^f(t)$ and $(t^a-t^b) \in P^f(a,b)$.

We want to characterize a neutral social choice function by the properties of its $P$-sets.
Here is a necessary and sufficient condition.
\begin{prop}[Neutrality]\label{prop:neu1}
Suppose $f$ is an implementable social choice function.
The social choice function $f$ is neutral if and only if
$P^f(a,b)=P^f(c,d)$ for all $a,b,c,d \in A$.
\end{prop}
\begin{proof}
Suppose $f$ is implementable and neutral.
Let $\alpha \in P^f(a,b)$. So, for some type profile $t$, we have $a \in C^f(t)$
and $t^a-t^b=\alpha$. Now, permuting $a,b$ respectively with $c,d$, we get
a new type profile $s$ with $s^c=t^a, s^d=t^b,s^a=t^c,s^b=t^d$. By neutrality,
$c \in C^f(s)$ and $s^c-s^d=t^a-t^b=\alpha$. So, $\alpha \in P^f(c,d)$.
Exchanging the role of $(a,b)$ and $(c,d)$ in this argument, we get that
$\alpha \in P^f(c,d)$ implies $\alpha \in P^f(a,b)$. Thus, $P^f(a,b)=P^f(c,d)$.

Now, suppose that $f$ is implementable and $P^f(a,b)=P^f(c,d)$ for all
$a,b,c,d \in A$. Consider a permutation $\rho$ of $A$. Without loss
of generality, assume that $\rho$ is a transposition, i.e., for some $a,b \in A$
we have $\rho(a)=b, \rho(b)=a$, and $\rho(c)=c$ for all $c \notin \{a,b\}$.
Consider a type profile $t \in \mathbb{T}^n$ and let $s$ be the type
profile induced by permutation $\rho$ on $t$, i.e., $s^a=t^b,s^b=t^a$,
and $s^{-ab}=t^{-ab}$. We show $f$ is neutral in several steps. \\

\noindent {\sc Step 1:} Suppose $a \notin C^f(t)$. We show that
$b \notin C^f(s)$. Assume for contradiction $b \in C^f(s)$. 
Let $c \in C^f(t)$. Such a $c$ exists since $C^f(t)$ is non-empty. Note that
$c \ne a$. There are two cases to consider. \\

\noindent {\sc Case 1:} Suppose $c=b$. Because, $b \in C^f(s)$,
we get that $(t^a-t^c) \in P^f(b,a)=P^f(a,c)$. \\

\noindent {\sc Case 2:} Suppose $c \notin \{a,b\}$. Again, because $b \in C^f(s)$,
we get that $(t^a-t^c) \in P^f(b,c)=P^f(a,c)$. \\ 

So, we get $(t^a-t^c) \in P^f(a,c)$ in both the cases.
Then for some $\epsilon \in \mathbb{R}^n$ and some type profile $v=(v^a=t^a+\epsilon,
v^c=t^c+\epsilon,v^{-ac})$, we have $a \in C^f(v)$. Consider the type profile
$u$ such that $u^a=t^a,u^c=t^c$, and $u^d=v^d-\epsilon$ for all $d \notin \{a,c\}$.
By PAD, $a \in C^f(u)$. But, in both $t$ and $u$, the utility vectors corresponding to
$a$ and $c$ are respectively $t^a$ and $t^c$. Since $a \notin C^f(t)$ and $c \in C^f(t)$,
by Proposition \ref{lem:new3}, $a \notin C^f(u)$. This is a contradiction. \\

\noindent {\sc Step 2:} Suppose $a \in C^f(t)$. We show that
$b \in C^f(s)$. Assume for contradiction $b \notin C^f(s)$. By
Step 1, $a \notin C^f(t)$. This is a contradiction. \\

\noindent {\sc Step 3:} Suppose $c \in C^f(t)$, where $c \notin \{a,b\}$. 
We show that $c \in C^f(s)$. Since $c \in C^f(t)$, we have
$(t^c-t^a),(t^c-t^b) \in P^f(c,b)$. Assume for contradiction $c \notin C^f(s)$.
Then, for some $d \ne c$, we have $d \in C^f(s)$.
There are two cases to consider. \\

\noindent {\sc Case 1:} Suppose $d \notin \{a,b,c\}$. 
In that case, by Proposition \ref{lem:new3} (applied to $s$ and $t$), $c \notin C^f(t)$.
This is a contradiction. \\

\noindent {\sc Case 2:} Suppose $d \in \{a,b\}$. Without
loss of generality, let $d=a$. So, $a \in C^f(s)$ but $c \notin C^f(s)$.
Now, since $(t^c-t^b) \in P^f(c,b) = P^f(c,a)$, there exists
a type profile $u=(u^a=t^b+\epsilon,u^c=t^c+\epsilon,u^{-ac})$
such that $c \in C^f(u)$. By PAD, $c \in C^f(v)$, where
$v^a=t^b, v^c=t^c$, and $v^d=u^d - \epsilon$ for all $d \notin \{a,c\}$.
By Proposition \ref{lem:new3}, we get that if $c \notin C^f(s)$, then
$a \notin C^f(s)$. This is a contradiction. \\

\noindent {\sc Step 4:} Suppose $c \notin C^f(t)$. Assume for contradiction
$c \in C^f(s)$. Exchanging the role of $s$ and $t$ in Step 3, we get that
$c \in C^f(t)$. This is a contradiction. \\

\noindent Combining all the steps, we get that $C^f(s)=\{\rho(c): c \in C^f(t)\}$,
i.e., $f$ is neutral.
\end{proof}

We begin by noting two properties of the $P$-sets. Identical properties have been established in
\cite{Lavi08,Vohra08} for their version of $P$-sets. We give proofs which are also
more direct.

\begin{lemma}\label{lem:psetlavi1}
Suppose $f$ is implementable and satisfies non-imposition. The following statements
are true for every $a,b,c \in A$.
\begin{enumerate}
\item If $(\beta - \epsilon) \in P^f(a,b)$ for some $\beta \in \mathbb{R}^n$ and some $\epsilon \in \mathbb{R}^n_{++}$,
then $-\beta \notin P^f(b,a)$.

\item If $\beta \in P^f(a,b)$ and $\alpha \in P^f(b,c)$, then $(\beta+\alpha) \in P^f(a,c)$.

\end{enumerate}
\end{lemma}
\begin{proof}
Fix $a,b,c \in A$. \\

\noindent {\sc Proof of (1):} Suppose $(\beta-\epsilon) \in P^f(a,b)$ for some $\beta \in \mathbb{R}^n$
and some $\epsilon \in \mathbb{R}^n_{++}$. Assume for contradiction that $-\beta \in P^f(b,a)$. So, there
exists some type profile $t$ such that $b \in C^f(t)$ and $t^a-t^b=\beta$.
Consider the type profile $s$ such that $s^a=t^a-\epsilon$ and $s^c=t^c$ for all $c \ne a$.
Note that $(s^a-s^b)=(\beta-\epsilon)$.
We first show that $a \in C^f(s)$. Since $(\beta-\epsilon) \in P^f(a,b)$, there is
some profile $u=(u^a=s^a+\alpha,u^b=s^b+\alpha,u^{-ab})$, where
$\alpha \in \mathbb{R}^n$, such that $a \in C^f(u)$.
By PAD, there is a profile $v=(v^a=s^a,v^b=s^b,v^{-ab})$ such that $a \in C^f(v)$.
We consider two cases. \\

\noindent {\sc Case 1:} Suppose $b \notin C^f(v)$. Then, by Proposition \ref{lem:new3}, $b \notin C^f(s)$.
By PAD, $b \notin C^f(t)$, which is a contradiction. \\

\noindent {\sc Case 1:} Suppose $b \in C^f(v)$. Then, by Proposition \ref{lem:new3}, $a \in C^f(s)$ if and only if
$b \in C^f(s)$. If $b \notin C^f(s)$, as in Case 1, we have a contradiction due to PAD. Hence, $a,b \in C^f(s)$.
Consider the type profile $x$ such that $x^a=t^a$, $x^b=t^b+\frac{\epsilon}{2}$, and $x^c=t^c$ for all $c \notin \{a,b\}$.
By PAD, $f(x)=a$. Hence, $b \notin C^f(t)$. This is a contradiction. \\

\noindent {\sc Proof of (2):} Suppose $\beta \in P^f(a,b)$ and $\alpha \in P^f(b,c)$.
Then, there must exist $t \in \mathbb{T}^n$ such that $a \in C^f(t)$ and $t^a-t^b=\beta$.
Now, consider a type profile $s$ such that $s^a=t^a$, $s^b=t^b$, $s^c=t^b-\alpha$, and
$s^d$ is sufficiently low for all $d \notin \{a,b,c\}$. We show that for all $d \notin \{a,b,c\}$,
we have $d \notin C^f(s)$. Assume for contradiction $d \in C^f(s)$. Then, by PAD, $a \notin C^f(t)$,
which is a contradiction. So, $C^f(s) \subseteq \{a,b,c\}$. 

We show that $a \in C^f(s)$. Assume for contradiction $a \notin C^f(s)$. Then, by Proposition \ref{lem:new3},
$b \notin C^f(s)$. This implies that $C^f(s)=\{c\}$ (by Lemma \ref{lem:new1}). 
By Lemma \ref{lem:ch1}, $(-\alpha - \epsilon) \in P^f(c,b)$.
By (1), $\alpha \notin P^f(b,c)$. This is a contradiction.

This implies that $a \in C^f(s)$, and hence, $(s^a-s^c) \in P^f(a,c)$. But
$s^a-s^c=t^a-t^b+\alpha=\beta+\alpha$ implies that $(\beta+\alpha) \in P^f(a,c)$.
\end{proof}

We are now ready to define the mapping that will make any social
choice function neutral.
Define the following mapping $\kappa:A \rightarrow \mathbb{R}$ as follows.
For all $a \in C^f(0)$~\footnote{Here, $0$ denotes the type profile, where
every agent's type is the $m$-dimensional zero vector.}, 
let $\kappa(a)=0$. For all $a \notin C^f(0)$, define
$\kappa(a)$ as follows. Denote a type vector $t$ as $1^b_{\epsilon}$,
where all utility (column)
vectors except one, say $t^b$, is zero vector and $t^b=1_{\epsilon}$ for
some $\epsilon \in \mathbb{R}$. For all $a \notin C^f(0)$,
\begin{align*}
\kappa(a) &= \{\epsilon \in \mathbb{R}_+: C^f(1^a_{\epsilon}) = C^f(0) \cup \{a\}\}.
\end{align*}

Our first claim is that for all $a \in A$, $\kappa(a) \in \mathbb{R}_+$ exists.

\begin{lemma}\label{lem:finitekappa}
Suppose $f$ is implementable and satisfies non-imposition.
Then, for all $a \in A$, $\kappa(a) \in \mathbb{R}_+$ and is unique.
Moreover, $\kappa(a) = \inf \{ \epsilon \in \mathbb{R}_+: a \in C^f(1^a_{\epsilon})\}$.
\end{lemma}
\begin{proof}
For all $a \in C^f(0)$, $\kappa(a)=0$, and hence, the lemma is true.
Consider $a \notin C^f(0)$. If $\kappa(a)$ exists, by PAD, it is unique.
We show that $\kappa(a)$ exists. We do this in two steps. \\

\noindent {\sc Step 1:} We show that there exists an $\epsilon \in \mathbb{R}_+$
such that $a \in C^f(1^a_{\epsilon})$. By non-imposition, there exists a type profile
$t$ such that $f(t)=a$. By PAD, there exists an $\epsilon \in \mathbb{R}$
such that $a \in C^f(1^a_{\epsilon})$. Moreover $\epsilon > 0$ since $a \notin C^f(0)$. \\

\noindent {\sc Step 2:} We now prove the lemma. Define
\begin{align*}
\kappa(a) &= \inf \{\epsilon: a \in C^f(1^a_{\epsilon})\}.
\end{align*}
By Step 1, $\kappa(a)$ exists. We show that $C^f(1^a_{\kappa(a)})=C^f(0) \cup \{a\}$.
Consider $b \notin (C^f(0) \cup \{a\})$. By PAD, if $b \in C^f(1^a_{\kappa(a)})$, then
$b \in C^f(0)$, which is a contradiction. Hence, $b \notin C^f(1^a_{\kappa(a)})$.
Next, by Proposition \ref{lem:new3}, we can conclude that either $C^f(1^a_{\kappa(a)})=C^f(0) \cup \{a\}$
or $C^f(1^a_{\kappa(a)}) = \{a\}$. Assume for contradiction $C^f(1^a_{\kappa(a)})=\{a\}$.
Then, by Lemma \ref{lem:ch1}, there exists $\epsilon \in \mathbb{R}^n_{++}$ such that
$a \in C^f(1^a_{\kappa(a)-\epsilon})$. This is a contradiction by the definition
of $\kappa(a)$. This shows that $C^f(1^a_{\kappa(a)})=C^f(0) \cup \{a\}$.
\end{proof}

We now prove a critical lemma.

\begin{lemma}\label{lem:zero}
Suppose $f$ is implementable and satisfies non-imposition. Let $t$ be
a type profile such that $t^a=1_{\kappa(a)}$ for all $a \in A$.
Then, $C^f(t)=A$.
\end{lemma}
\begin{proof}
We start from the type profile $0$ and move to $t$ in finite number of
steps. Consider a set $A^0 \subseteq A$. 
Initially, $A^0=A \setminus C^f(0)$.
Now, choose $a \in A^0$, and consider $1^a_{\kappa(a)}$.
By definition of $\kappa(a)$, $C^f(1^a_{\kappa(a)})=\{a\} \cup C^f(0)$.
Now, set $A^0:= A^0 \setminus \{a\}$, and choose
$b \in A^0$. We now define a type profile $s$ such that $s^a=1_{\kappa(a)}$
and $s^b=1_{\kappa(b)}$ but $s^c=0$ for all $c \notin \{a,b\}$.
By Proposition \ref{lem:new3}, either $C^f(s)=C^f(1^a_{\kappa(a)}) \cup \{b\}$
or $C^f(s)=\{b\}$. The latter case is not possible by Lemma \ref{lem:ch1}
since it will imply $b \in C^f(1^b_{\kappa(b)-\epsilon})$ for some $\epsilon \in \mathbb{R}^n_{++}$, 
which will violate the definition of $\kappa(b)$.
Hence, $C^f(s)=C^f(1^a_{\kappa(a)}) \cup \{b\}$.
Now, we set $A^0:= A^0 \setminus \{b\}$, and repeat. Since $A$ is finite,
this process will terminate with type profile $t$ such that $C^f(t)=A$.
\end{proof}

We now have all ingredients for proving the Roberts' theorem.
\begin{theorem}[\cite{Roberts79}]\label{theo:robertsnewproof}
Suppose $\mathbb{T}^n=\mathbb{R}^{m \times n}$. If $f$ is an implementable social choice function 
and satisfies non-imposition, then there exists weights $\lambda \in \mathbb{R}^n_+ \setminus \{0\}$ and 
a deterministic real-valued function $\kappa:A \rightarrow \mathbb{R}$ such that for all $t \in \mathbb{T}^n$,
\begin{align*}
f(t) &\in \arg \max_{a \in A}\big[\sum_{i \in N}\lambda_it_i^a-\kappa(a)\big]
\end{align*}
\end{theorem}
\begin{proof}
Since $f$ is implementable and satisfies non-imposition, by Lemma \ref{lem:finitekappa},
there exists a mapping $\kappa:A \rightarrow \mathbb{R}$ satisfying properties
stated in Lemmas \ref{lem:finitekappa} and \ref{lem:zero}. Now, consider the social
choice function $f^{\kappa}$. By Proposition \ref{prop:unr1}, $f^{\kappa}$ is implementable.
By definition,
$f^{\kappa}(0)=f(1_{\kappa})$. By Lemma \ref{lem:zero}, $C^{f^{\kappa}}(0)=C^f(t)=A$. This implies
that $0 \in P^{f^{\kappa}}(a,b)$ for all $a,b \in A$.

Now, pick $a,b,c,d \in A$ and let $\beta \in P^{f^{\kappa}}(a,b)$. But $0 \in P^{f^{\kappa}}(b,d)$.
By Lemma \ref{lem:psetlavi1}, $\beta \in P^{f^{\kappa}}(a,d)$. Now, using $0 \in P^{f^{\kappa}}(c,a)$,
and applying Lemma \ref{lem:psetlavi1} again, we get $\beta \in P^{f^{\kappa}}(c,d)$.
By Proposition \ref{prop:neu1}, $f^{\kappa}$ is neutral.
By Theorem \ref{theo:newroberts}, $f^{\kappa}$ is a weighted
welfare maximizer. This implies that there exists $\lambda \in \mathbb{R}^n \setminus
\{0\}$ such that for every $t \in \mathbb{T}^n$,
\begin{align*}
f^{\kappa}(t) &\in \arg \max_{a \in A}\sum_{i=1}^n\lambda_i t^a_i.
\end{align*}
But this implies that, for every $t \in \mathbb{T}^n$,
\begin{align*}
f^{\kappa}(t-1_{\kappa}) &\in \arg \max_{a \in A}\sum_{i=1}^n\lambda_i(t-1_{\kappa})^a_i.
\end{align*}
This in turn implies that, for every $t \in \mathbb{T}^n$,
\begin{align*}
f(t) &\in \arg \max_{a \in A}\big[\sum_{i=1}^n\lambda_i [t^a_i - \kappa(a)]\big]
\end{align*}

Since we can assume without loss of generality that $\lambda_i \in [0,1]$ for all $i \in N$,
we can immediately infer Roberts' theorem.
\end{proof}

To summarize, Roberts' theorem can be proved using Roberts' theorem with neutrality by
transforming any social choice function to a neutral social choice function as given by
Proposition \ref{prop:neu1}. This transformation seems to require that the domain be
unrestricted.

\section{Conclusion}
 
 We have provided a characterization of domains over which every implementable and
 neutral social choice function is a weighted welfare maximizer. Our proof technique
 reduces the problem of characterizing such social choice functions to the
 problem of characterizing orderings over Euclidean space, a problem which has
 been studied at length in social choice theory. Finally, we show how Roberts' theorem
 (the general version) can be proved using Roberts' theorem with neutrality. This
 proof requires transforming any implementable social choice function into a neutral
 and implementable social choice function. To our knowledge, this transformation seems
 to require the unrestricted domain.

We summarize our main contribution in Figure \ref{fig:contr}. The arrows in this figure indicate
implications. As the figure shows, our results can be thought to be equivalence of the
PAD condition and implementability in the presence of neutrality in open interval domains. 
It will be interesting to investigate this
equivalence in the absence of neutrality.
 
\begin{figure}[!tbh]
\centering
\setlength{\unitlength}{3118sp}%
\begingroup\makeatletter\ifx\SetFigFont\undefined%
\gdef\SetFigFont#1#2#3#4#5{%
  \reset@font\fontsize{#1}{#2pt}%
  \fontfamily{#3}\fontseries{#4}\fontshape{#5}%
  \selectfont}%
\fi\endgroup%
\begin{picture}(8349,5016)(5164,-6115)
\thinlines
{\color[rgb]{0,0,0}\put(8731,-1606){\oval(210,210)[bl]}
\put(8731,-1216){\oval(210,210)[tl]}
\put(10021,-1606){\oval(210,210)[br]}
\put(10021,-1216){\oval(210,210)[tr]}
\put(8731,-1711){\line( 1, 0){1290}}
\put(8731,-1111){\line( 1, 0){1290}}
\put(8626,-1606){\line( 0, 1){390}}
\put(10126,-1606){\line( 0, 1){390}}
}%
{\color[rgb]{0,0,0}\put(12106,-1606){\oval(210,210)[bl]}
\put(12106,-1216){\oval(210,210)[tl]}
\put(13396,-1606){\oval(210,210)[br]}
\put(13396,-1216){\oval(210,210)[tr]}
\put(12106,-1711){\line( 1, 0){1290}}
\put(12106,-1111){\line( 1, 0){1290}}
\put(12001,-1606){\line( 0, 1){390}}
\put(13501,-1606){\line( 0, 1){390}}
}%
{\color[rgb]{0,0,0}\put(8731,-5656){\oval(210,210)[bl]}
\put(8731,-5266){\oval(210,210)[tl]}
\put(10021,-5656){\oval(210,210)[br]}
\put(10021,-5266){\oval(210,210)[tr]}
\put(8731,-5761){\line( 1, 0){1290}}
\put(8731,-5161){\line( 1, 0){1290}}
\put(8626,-5656){\line( 0, 1){390}}
\put(10126,-5656){\line( 0, 1){390}}
}%
{\color[rgb]{0,0,0}\put(12106,-3781){\oval(210,210)[bl]}
\put(12106,-3391){\oval(210,210)[tl]}
\put(13396,-3781){\oval(210,210)[br]}
\put(13396,-3391){\oval(210,210)[tr]}
\put(12106,-3886){\line( 1, 0){1290}}
\put(12106,-3286){\line( 1, 0){1290}}
\put(12001,-3781){\line( 0, 1){390}}
\put(13501,-3781){\line( 0, 1){390}}
}%
{\color[rgb]{0,0,0}\put(5281,-1606){\oval(210,210)[bl]}
\put(5281,-1216){\oval(210,210)[tl]}
\put(6571,-1606){\oval(210,210)[br]}
\put(6571,-1216){\oval(210,210)[tr]}
\put(5281,-1711){\line( 1, 0){1290}}
\put(5281,-1111){\line( 1, 0){1290}}
\put(5176,-1606){\line( 0, 1){390}}
\put(6676,-1606){\line( 0, 1){390}}
}%
\thicklines
{\color[rgb]{0,0,0}\put(6676,-1261){\vector( 1, 0){1950}}
}%
{\color[rgb]{0,0,0}\put(8626,-1486){\vector(-1, 0){1950}}
}%
{\color[rgb]{0,0,0}\put(10126,-1261){\vector( 1, 0){1800}}
}%
{\color[rgb]{0,0,0}\put(12001,-1486){\vector(-1, 0){1875}}
}%
{\color[rgb]{0,0,0}\put(12751,-1711){\vector( 0,-1){1575}}
}%
\thinlines
{\color[rgb]{0,0,0}\put(8731,-3781){\oval(210,210)[bl]}
\put(8731,-3391){\oval(210,210)[tl]}
\put(10021,-3781){\oval(210,210)[br]}
\put(10021,-3391){\oval(210,210)[tr]}
\put(8731,-3886){\line( 1, 0){1290}}
\put(8731,-3286){\line( 1, 0){1290}}
\put(8626,-3781){\line( 0, 1){390}}
\put(10126,-3781){\line( 0, 1){390}}
}%
\thicklines
{\color[rgb]{0,0,0}\put(12751,-3886){\line( 0,-1){1500}}
\put(12751,-5386){\vector(-1, 0){2625}}
}%
{\color[rgb]{0,0,0}\put(8626,-3586){\line(-1, 0){2175}}
\put(6451,-3586){\vector( 0, 1){1875}}
}%
{\color[rgb]{0,0,0}\put(8626,-5461){\line(-1, 0){2925}}
\put(5701,-5461){\vector( 0, 1){3750}}
}%
{\color[rgb]{0,0,0}\put(12001,-3586){\vector(-1, 0){1875}}
}%
\put(8851,-1591){\makebox(0,0)[lb]{\smash{{\SetFigFont{9}{10.8}{\familydefault}{\mddefault}{\updefault}{\color[rgb]{0,0,0}Monotonicity}%
}}}}
\put(12226,-1591){\makebox(0,0)[lb]{\smash{{\SetFigFont{9}{10.8}{\familydefault}{\mddefault}{\updefault}{\color[rgb]{0,0,0}Monotonicity}%
}}}}
\put(12601,-3661){\makebox(0,0)[lb]{\smash{{\SetFigFont{9}{10.8}{\familydefault}{\mddefault}{\updefault}{\color[rgb]{0,0,0}PAD}%
}}}}
\put(9076,-1336){\makebox(0,0)[lb]{\smash{{\SetFigFont{9}{10.8}{\familydefault}{\mddefault}{\updefault}{\color[rgb]{0,0,0}Cycle}%
}}}}
\put(12526,-1336){\makebox(0,0)[lb]{\smash{{\SetFigFont{9}{10.8}{\familydefault}{\mddefault}{\updefault}{\color[rgb]{0,0,0}Weak}%
}}}}
\put(8926,-5386){\makebox(0,0)[lb]{\smash{{\SetFigFont{9}{10.8}{\familydefault}{\mddefault}{\updefault}{\color[rgb]{0,0,0}Weighted}%
}}}}
\put(8926,-5611){\makebox(0,0)[lb]{\smash{{\SetFigFont{9}{10.8}{\familydefault}{\mddefault}{\updefault}{\color[rgb]{0,0,0}Efficiency}%
}}}}
\put(9076,-3511){\makebox(0,0)[lb]{\smash{{\SetFigFont{9}{10.8}{\familydefault}{\mddefault}{\updefault}{\color[rgb]{0,0,0}Affine}%
}}}}
\put(8926,-3736){\makebox(0,0)[lb]{\smash{{\SetFigFont{9}{10.8}{\familydefault}{\mddefault}{\updefault}{\color[rgb]{0,0,0}Maximizers}%
}}}}
\put(10576,-3811){\makebox(0,0)[lb]{\smash{{\SetFigFont{9}{10.8}{\familydefault}{\mddefault}{\updefault}{\color[rgb]{0,0,0}Unrestricted}%
}}}}
\put(10576,-4036){\makebox(0,0)[lb]{\smash{{\SetFigFont{9}{10.8}{\familydefault}{\mddefault}{\updefault}{\color[rgb]{0,0,0}Domains}%
}}}}
\put(10576,-4261){\makebox(0,0)[lb]{\smash{{\SetFigFont{9}{10.8}{\familydefault}{\mddefault}{\updefault}{\color[rgb]{0,0,0}\cite{Roberts79}}%
}}}}
\put(10651,-5611){\makebox(0,0)[lb]{\smash{{\SetFigFont{9}{10.8}{\familydefault}{\mddefault}{\updefault}{\color[rgb]{0,0,0}Neutrality and}%
}}}}
\put(12901,-2386){\makebox(0,0)[lb]{\smash{{\SetFigFont{9}{10.8}{\familydefault}{\mddefault}{\updefault}{\color[rgb]{0,0,0}\cite{Roberts79}}%
}}}}
\put(7276,-1711){\makebox(0,0)[lb]{\smash{{\SetFigFont{9}{10.8}{\familydefault}{\mddefault}{\updefault}{\color[rgb]{0,0,0}\cite{Rochet87}}%
}}}}
\put(10201,-1786){\makebox(0,0)[lb]{\smash{{\SetFigFont{9}{10.8}{\familydefault}{\mddefault}{\updefault}{\color[rgb]{0,0,0}Auction Domains}%
}}}}
\put(10201,-2011){\makebox(0,0)[lb]{\smash{{\SetFigFont{9}{10.8}{\familydefault}{\mddefault}{\updefault}{\color[rgb]{0,0,0}\cite{Bikh06}}%
}}}}
\put(10651,-6061){\makebox(0,0)[lb]{\smash{{\SetFigFont{9}{10.8}{\familydefault}{\mddefault}{\updefault}{\color[rgb]{0,0,0}This paper}%
}}}}
\put(10651,-5836){\makebox(0,0)[lb]{\smash{{\SetFigFont{9}{10.8}{\familydefault}{\mddefault}{\updefault}{\color[rgb]{0,0,0}Open Interval Domains}%
}}}}
\put(5251,-1486){\makebox(0,0)[lb]{\smash{{\SetFigFont{9}{10.8}{\familydefault}{\mddefault}{\updefault}{\color[rgb]{0,0,0}Implementability}%
}}}}
\end{picture}%
\caption{Understanding Implementability}
\label{fig:contr}
\end{figure}
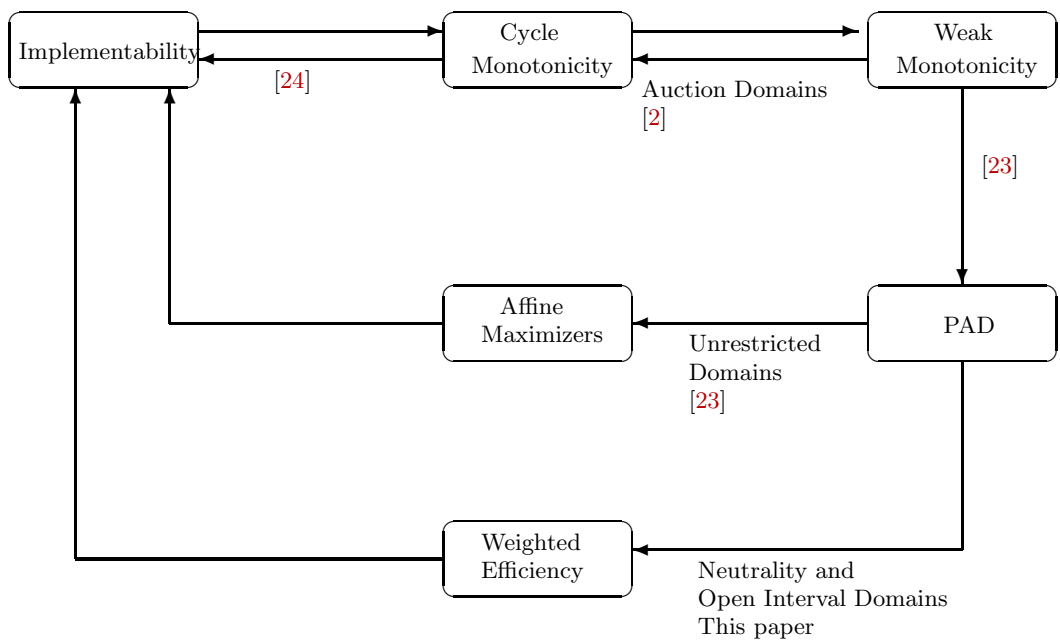

\section*{Appendix A}\label{sec:appa}

\noindent {\sc Proof of Fact \ref{fact:f1}}

\begin{proof}
Let $x,y \in X$ and
$z=\alpha x + (1-\alpha) y$ for some $\alpha \in (0,1)$.
We consider two possible cases. \\

\noindent {\sc Case 1:} Suppose $X$ is closed in $\mathbb{D}$.
Assume for contradiction that $z \notin X$. Since $D$ is
convex, $z \in \mathbb{D} \setminus X$. Since $X$ is closed
in $\mathbb{D}$, the set $\mathbb{D} \setminus X$ is open
in $\mathbb{D}$. Hence, $\mathbb{D} \setminus X$ is open in
$\mathbb{R}^n$. This means, there exists an $n$-dimensional
open ball $B_{\delta}(z)=\{z': \lVert z'-z \rVert < \delta\}$ of radius $\delta$ such that
every $z' \in B_{\delta}(z)$ belongs to $\mathbb{D} \setminus X$.

Now, consider an iterative procedure as follows. Let $l,h$ be two variables
in $\mathbb{R}^n$. Initially, set $l=x$ and $h=y$. In every step,
\begin{itemize}
\item if $z$ is in the convex hull of $l$ and $\frac{l+h}{2}$ then set
$h=\frac{l+h}{2}$,
\item else set $l=\frac{l+h}{2}$.
\end{itemize}
If $\lVert l-h \rVert < 2\delta$, stop. Else, repeat the step.

Since $\lVert l-h \rVert$ strictly decreases in every step, the procedure
will terminate. Moreover, $l$ and $h$ at the end of the procedures are
two points in $X$. Hence, $\frac{l+h}{2}$ is in $X$ and lies in the ball
$B_{\delta}(z)$. This is a contradiction~\footnote{Essentially, the procedure
generates a sequence of dyadic rational numbers. We know that the set
of dyadic rational numbers are dense. Since $X$ is closed, we are done.}. \\

\noindent {\sc Case 2:} Suppose $X$ is open in $\mathbb{D}$.
Then $X$ is open in $\mathbb{R}^n$. This implies that there exists
an open ball $B_{\delta_x}(x)$ around $x$ of radius $\delta_x$ and an
open ball $B_{\delta_y}(y)$ around $y$ of radius $\delta_y$ such that
each of these balls are contained in $X$. Let $\delta=\min(\delta_x,\delta_y)$.
Using the fact that for every $x' \in B_{\delta_x}(x)$ and every $y' \in B_{\delta_y}(y)$
we have $\frac{x'+y'}{2} \in X$, we get that every $x'' \in B_{\delta}(\frac{x+y}{2})$
lies in $X$. Now, we can repeat the procedure of Case 1 to conclude that
$z \in X$.
\end{proof}

\noindent {\sc Proof of Theorem \ref{theo:efficiency}} \\

\begin{proof}
Note that $\mathbb{D}$ is open from above (i.e., for every $x \in D$,
there exists an $\epsilon \in \mathbb{R}^n_{++}$ such that $(x+\epsilon) \in D$) and 
a meet-semilattice (i.e., if $x,y \in D$, then $\min(x,y) \in D$). We can verify that
Propositions \ref{prop:order} and \ref{prop:axioms} are true as long as
$\mathbb{D}$ is open from above and a meet-semilattice. Hence,
by Proposition \ref{prop:order}, $R^f$ is an ordering. By Proposition \ref{prop:axioms}
and Lemma \ref{lem:anon}, $f$ satisfies weak Pareto, invariance, and anonymity (we do
not need continuity for this proof). Also, note that for any $x \in \mathbb{D}$, any permutation
of the elements of $x$ results in a vector in $\mathbb{D}$.

Now, choose $x,y \in \mathbb{D}$ such that $\sum_{i \in N}x_i=\sum_{i \in N}y_i$.
By anonymity, we can rearrange $x$ and $y$ in non-decreasing order but mutually
ranked the same way as $x$ and $y$. Considering successively, in these new vectors,
each pair of corresponding components and subtracting from each the minimal one,
we get again two new vectors which are ranked the same way as $x$ and $y$ by invariance
(note here that these two new vectors belong to $\mathbb{D}=[0,H)^n$). Repeating these
two operations at most $n$ times, we will reach two zero vectors (since $\sum_{i \in N}x_i=
\sum_{i \in N}y_i$). Hence, $x I^f y$.

Next, we show that if $\sum_{i \in N}x_i > \sum_{i \in N}y_i$ then $x P^f y$.
Let $\delta=\frac{1}{n}[\sum_{i \in N}x_i-\sum_{i \in N}y_i]$. Consider the vector
$z$ defined as $z_i=y_i+\delta$ for all $i \in N$. By weak Pareto $z P^f y$.
Further $\sum_{i \in N}x_i=\sum_{i \in N}z_i$. Hence, $x I^f z$. Hence, $x P^f y$.

By Lemma \ref{lem:new1}, for every $t \in \mathbb{T}^n$, we have
$f(t) \in C^f(t)$. Hence, $f(t) R^f a$ for all $a \in A$. Hence, $f$ is the
efficient social choice function.
\end{proof}

\section*{Appendix B}\label{sec:appb}

In this appendix, we show that a stronger, but natural definition of neutrality
implies our definition of neutrality.

\begin{defn}
A social choice function $f$ is {\bf scf-neutral} if for every $t \in \mathbb{T}^n$,
every permutation $\rho$ of $A$ and type profile $s$ induced by permutation $\rho$
on $t$, we have $f(s)=\rho(f(t))$ if $t \ne s$.
\end{defn}

\begin{claim}\label{cl:neut}
If a social choice function $f$ is implementable and scf-neutral, then it
is neutral.
\end{claim}
\begin{proof}
Since $f$ is implementable, it satisfies PAD. Fix a type profile $t$ and a permutation $\rho$
of $A$, and let $s$ be the type profile induced by $\rho$ on $t$. Consider $a \in C^f(t)$
and a type profile $u=(u^a=t^a+\epsilon,u^{-a}=t^{-a})$ for some $\epsilon \in \mathbb{R}^n_{++}$.
Hence, $f(u)=a$. Now, let $v$ be the type profile induced by permutation $\rho$ on
$u$. By scf-neutrality $f(v)=\rho(a)$ ($\epsilon$ can be chosen arbitrarily small so that
$u \ne v$). By PAD, $\rho(a) \in C^f(s)$.

To show that for any $a \notin C^f(t)$, we must have $a \notin C^f(s)$, assume for contradiction
$a \in C^f(s)$, and apply the previous argument to conclude $a \in C^f(t)$. This gives the
desired contradiction.
\end{proof}

\end{document}